\documentclass[journal,10pt,twocolumn,letter]{IEEEtran}

\usepackage{graphics}
\usepackage{colortbl}
\usepackage{multicol}
\usepackage{lipsum}
\usepackage{color}
\usepackage[cmex10]{amsmath}
\usepackage{amsthm}
\usepackage{amssymb}
\usepackage{mathrsfs}
\usepackage{mathtools}
\usepackage{amsbsy}
\usepackage[colorlinks=true,bookmarks=false,citecolor=blue,urlcolor=blue]{hyperref}
\usepackage{xcolor}
\usepackage{mathrsfs}
\usepackage{setspace}
\usepackage{float}
\usepackage{graphicx}
\usepackage{pstool}
\usepackage{lettrine}
\usepackage{newclude}
\usepackage[normalem]{ulem}
\usepackage{latexsym}
\usepackage{algpseudocode}
\usepackage{gensymb}

\usepackage{balance}

\usepackage{multirow}
\usepackage{amsmath}
\usepackage{amsmath,bm}
\usepackage{wasysym}
\usepackage{mathrsfs}
\usepackage{amsmath,cite,amsfonts,amssymb,color}

\usepackage[linesnumbered,ruled,vlined]{algorithm2e}
\SetKwInput{KwInput}{Input}  
\SetKwInput{KwOutput}{Output}
\SetKw{KwBy}{by}
\makeatletter
\newcommand{\nosemic}{\renewcommand{\@endalgocfline}{\relax}}
\newcommand{\dosemic}{\renewcommand{\@endalgocfline}{\algocf@endline}}
\let\oldnl\nl
\newcommand{\nonl}{\renewcommand{\nl}{\let\nl\oldnl}}
\makeatother
\ifCLASSOPTIONcompsoc
\usepackage[caption=false,font=normalsize,labelfon
t=sf,textfont=sf]{subfig}
\else
\usepackage[caption=false,font=footnotesize]{subfig}
\fi

\DeclarePairedDelimiterX\MeijerM[3]{\lparen\!}{\rparen}%
{\,#3\delimsize\vert\begin{smallmatrix}#1 \\ #2\end{smallmatrix}}

\newcommand\MeijerG[8][]{%
  G^{\,#2,#3}_{#4,#5}\MeijerM[#1]{#6}{#7}{#8}}

\WithSuffix\newcommand\MeijerG*[7]{%
  G^{\,#1,#2}_{#3,#4}\MeijerM*{#5}{#6}{#7}}

\allowdisplaybreaks

\newtheorem{theorem}{Theorem}

\graphicspath{{figures/}}
\allowdisplaybreaks

\newcommand{\RNum}[1]{\uppercase\expandafter{\romannumeral #1\relax}}

\usepackage{mathtools}

\usepackage{accents}
\newlength{\dhatheight}
\newlength{\dtildeheight}
\newcommand{\doublehat}[1]{%
    \settoheight{\dhatheight}{\ensuremath{\hat{#1}}}%
    \addtolength{\dhatheight}{-0.2ex}%
    \hat{\vphantom{\rule{1pt}{\dhatheight}}%
    \smash{\hat{#1}}}}

\begin{document}

\title{High-Speed Trains Access Connectivity Through RIS-Assisted FSO Communications}

\author{\IEEEauthorblockN{Pouya Agheli\thanks{The authors are with the Department of Electrical Engineering, Amirkabir University of Technology (Tehran Polytechnic), Tehran 1591634311, Iran (E-mails: \{pouya.agheli, beyranvand, mj.emadi\}@aut.ac.ir).},
Hamzeh Beyranvand,
and Mohammad Javad Emadi
}} 


\maketitle

\begin{abstract}
Free-space optic (FSO) is a promising solution to provide broadband Internet access for high-speed trains (HSTs). 
Besides, reconfigurable intelligent surfaces (RIS) are considered as hardware technology to improve performance of optical wireless communication systems. 
In this paper, we propose a RIS-assisted FSO system to provide access connectivity for HTSs, as an upgrade for the existing direct and relay-assisted FSO access setups.
Our motivation is mainly based on well-proven results indicating that a RIS-assisted optical wireless system, with a large enough number of RIS elements, outperforms a relay-assisted one thanks to its programmable structure. 
We firstly compute the statistical expressions of the considered RIS-assisted FSO channels under weak and moderate-to-strong fading conditions. Then, the network's average signal-to-noise ratio and outage probability are formulated based on the assumed fading conditions, and for two fixed- and dynamic-oriented RIS coverage scenarios. Our results reveal that the proposed access network offers up to around $44\%$ higher data rates and $240\%$ wider coverage area for each FSO base station (FSO-BS) compared to those of the relay-assisted one. The increase of coverage area, on average, reduces  $67\%$ the number of required FSO-BSs for a given distance, which results in fewer handover processes compared to the alternative setups. Finally, the results are verified through Monte-Carlo simulations.

\end{abstract}

\IEEEpeerreviewmaketitle

\section{Introduction}
\IEEEPARstart{T}{he} drastic increase of Internet-based applications and video on demand streaming along with the growing interest in high-speed trains (HSTs), has created a popular research topic on providing seamless high-capacity access networks for HSTs. The main objective is to ensure high quality of service (QoS), i.e., high data rates and low latency, for the passengers in an HST similar to generic pedestrian users despite the train's high speed of up to $500\,[\text{km/h}]$ \cite{kim2018comprehensive}. Indeed, we must propose a reliable access network, to not only enables high-rate communications but also results in low handover frequency \cite{fathi2017optimal}. To achieve this goal, multiple radio and optical wireless technologies with different setups have been utilized since the first launch of HSTs \cite{kim2018comprehensive,fathi2017optimal,khallaf2021comprehensive,ai2014challenges,goller1995application,wang2015channel,aguado2008wimax,kowal2010operational,paudel2013modelling,paudel2016laboratory,han2015radiate,taheri2017provisioning,fan2018reducing,fan2017reducing,kaymak2017divergence,mabrouk2019enhancement,mohan2020sectorised}.
The well-known wireless access solutions for HSTs, such as the global system for mobile communications-railway (GSM-R) \cite{goller1995application}, long-term evolution for railways (LTE-R) \cite{wang2015channel}, dual-hop worldwide interoperability for microwave access (WiMAX) \cite{aguado2008wimax}, and dual-hop mobile broadband wireless access (MBWA) \cite{kowal2010operational}, have been deployed in several countries. However, these technologies cannot ensure high QoS for upcoming HSTs in the presence of rapidly time-varying and non-stationary wireless channels, co-channel interference, high path-loss attenuation for millimeter-wave and terahertz bands, high handover frequency, and strong Doppler shifts \cite{khallaf2021comprehensive,ai2014challenges}. 
To address some fundamental challenges of the proposed radio technologies for HSTs, \emph{free-space optic} (FSO) has gained research attention, as discussed in \cite{fathi2017optimal}, \cite{khallaf2021comprehensive}, and
\cite{paudel2013modelling,paudel2016laboratory,han2015radiate,taheri2017provisioning,fan2018reducing,fan2017reducing,kaymak2017divergence,mabrouk2019enhancement,mohan2020sectorised}.

The FSO technology offers high capacity with rapid setup time, easy upgrade, flexibility, freedom from spectrum license regulations, and enhanced security \cite{abdalla2020optical}. However, it comes at the expense of some drawbacks such as pointing error, line-of-sight connectivity requirement, and sensitivity to atmospheric conditions \cite{khalighi2014survey,agheli2021uav,kaushal2016optical,agheli2021design}.
The FSO channel modeling in HSTs is faced with some challenges by reason of unsteady environmental differences, such as tunnels, hills, and dense areas. In this case, an analytical study on ground-to-train channel modeling for typical FSO links has been performed in \cite{paudel2013modelling}.  
An experimental prototype for the demonstration of ground-to-train FSO communications has been also presented in \cite{paudel2016laboratory}. 
Due to the large divergence angles of FSO beams for HST systems and as a result severe distance-dependent beam spreading, each FSO base station (FSO-BS) serves only a narrow region. Thus, we face some challenges such as demanding a large number of FSO-BSs, frequent handovers, and high capital expenditures (CAPEX) to provide seamless coverage, \cite{fathi2017optimal} and \cite{han2015radiate,taheri2017provisioning,fan2018reducing,fan2017reducing}.

Han et al. \cite{han2015radiate} have proposed a radio-over-fiber setup as antenna extender (RADIATE) solution to address the frequent handover and rapid channel variation issues. This solution provides broadband Internet services for HSTs with cellular backhaul networks. 
The authors of \cite{taheri2017provisioning} have utilized the advantages of the RADIATE solution and proposed an FSO-based communication technology for HSTs, referred to as free-space optic utilization in high-speed trains (FOCUS).
The FOCUS solution reduces frequent handovers, prevents handoff congestion, balances the traffic load among transceivers, and handles fast channel variations in HST communications.
In addition, Fan et al. \cite{fan2018reducing} have exploited the dual transceivers scheme in a ground-to-train communications system, where two transceivers of an FSO-BS can point to different directions. Thanks to this scheme, two transceivers at the front- and back-end of HST can cooperate to maintain continuous ground-to-train FSO communications. Therefore, the number of FSO-BSs decreases and the impacts of frequent handovers are alleviated.
The same authors have also suggested a rotating transceiver scheme to mitigate the impacts of handover processes via utilizing steerable FSO transceivers for HST communications, as in \cite{fan2017reducing}. 
Moreover, as an attempt to provide seamless FSO communications for next-generation HSTs, optimal positions of FSO-BSs have been found in \cite{fathi2017optimal} under two single- and dual-wavelength FSO coverage models.

In order to decrease the number of FSO-BSs and handover processes, an interesting solution is to take advantage of relay-assisted transmissions for FSO links. In this regard, multi-hop short-length FSO links would experience better conditions compared to that of the direct transmissions \cite{safari2008relay}. Thus, each FSO-BS with its paired relay can support a wider area. To choose an appropriate relaying model, it has been shown that all-optical relaying offers a favorable trade-off between complexity and performance, and it is considered as a low-complexity solution \cite{kashani2012all}.
Most recently, it has been shown that UAV-based FSO air-relay systems provide wider coverage areas for HSTs compared to conventional ground-relay ones \cite{khallaf2021comprehensive}. The UAV-assisted air-relay systems for HSTs were firstly introduced in \cite{khallaf2019uav}, where a flying UAV relays data between a moving train car and a network gateway. Even though the air-relay systems perform better than the fixed ground-relay systems, it is not affordable and optimal to build a UAV swarm over a moving train to provide Internet access for the train's passengers in non-emergency conditions. 

Recently, \emph{reconfigurable intelligent surfaces} (RISs) have been defined as hardware technology to improve the performances of wireless and optical wireless systems \cite{yuan2021ris,najafi2019intelligentold,najafi2019intelligent,ndjiongue2021analysis}.
The fundamental features and channel modeling of RIS-assisted optical wireless systems have been widely studied in the literature, e.g., \cite{najafi2019intelligentold,najafi2019intelligent}, and \cite{najafi2020physics,ajam2021channel,jia2020ergodic,yang2020mixed,abdelhady2020visible,ndjiongue2021re}, under different physical conditions.
It is well known that we can use two different types of optical RIS structures for deploying a RIS-assisted optical wireless system; \emph{metasurfaces} and \emph{mirror arrays} \cite{abdelhady2020visible}. These structures can be based on integrated seamless surfaces or hundreds of micro-mirrors built with micro-electromechanical systems (MEMS).
In dual-hop FSO systems, the RIS technology can gather dispersed beams and reconfigure them through desired directions, whereas the relay systems can only reflect a portion of the incident FSO beams and disregards the rest.
Indeed, it has been proven that RIS-assisted systems, with a large enough number of elements, outperform the relay ones for wireless transmissions \cite{di2020reconfigurable} and \cite{bjornson2019intelligent}. The results can be extended for optical wireless transmissions.
According to the mentioned factors, the RIS technology can add novelty to the existing direct or relay-assisted FSO access networks for HSTs and be considered as a promising solution to decrease the number of required FSO-BSs for a given distance.
In this paper, we analyze a RIS-assisted FSO system to provide an access network for an HST, which has not been studied to the best of our knowledge. In summary, the key contributions of the paper are as follows.
\begin{itemize}
    \item Deriving the channel probability distribution function (p.d.f) of the proposed RIS-assisted FSO system under weak and moderate-to-strong fading conditions.
    \item Formulating the average signal-to-noise ratio (SNR) and closed-form outage probability expressions under the assumed fading conditions and considering two scenarios for RIS coverage; fixed- and dynamic-oriented ones.
    \item Evaluating the proposed system's performance through Monte-Carlo simulations and analytical results. 
    Furthermore, it is shown that the proposed access network not only increases the achievable data rates at the HST but also boosts the coverage area supported by each FSO-BS and its paired RIS compared to those of the relay-assisted alternatives. Also, our proposed system results in lower number of FSO-BSs, handover processes, and CAPEX. 
    \item Presenting a system design framework for practical deployments.
\end{itemize}

\textit{Organization:} Section~\ref{Sec:Sec2} introduces the proposed RIS-assisted FSO system. The network's performance metrics are derived in Section~\ref{Sec:Sec3}. Numerical results and discussions in addition to the system design framework are presented in Section~\ref{Sec:Sec4}. Finally, the paper is concluded in Section~\ref{Sec:Sec5}.

\textit{Notation:} $\Gamma(.)$ and $\Gamma(s,x)\!=\!\int_{x}^{\infty}t^{s-1} e^{-t} dt$ sequentially denote the Gamma and upper incomplete Gamma functions. $K_\alpha(.)$ is the modified Bessel function of the second kind. Also, $\operatorname{erf}(x)\!=\!\frac{2}{\sqrt{\pi}}\int_{0}^{x}e^{-t^2}dt$ and $\operatorname{erfc}(x)\!=\!1\!-\operatorname{erf}(x)$ indicate the error and complementary error functions, respectively. Besides, $\mathbb{E}\{\cdot\}$ is the statistical expectation, and $z \!\sim\! \mathcal{N}(m,\sigma^2)$ shows real-valued symmetric Gaussian random variable with mean $m$ and variance $\sigma^2$. Moreover, $\mathbb{I}(\mathcal{C})$ denotes an indicator function, where $\mathbb{I}(\mathcal{C})\!=\!1$ if the condition $\mathcal{C}$ is satisfied, and $\mathbb{I}(\mathcal{C})\!=\!0$, otherwise. Finally, $\MeijerG*{m}{n}{p}{q}{a_1,a_2,...,a_p}{b_1,b_2,...,b_q}{z}$ presents the Meijer's G-function.

\section{System Model}\label{Sec:Sec2}
\begin{figure*}[h!]
    \centering
    \hspace{-0.9cm}
    \pstool[scale=0.9]{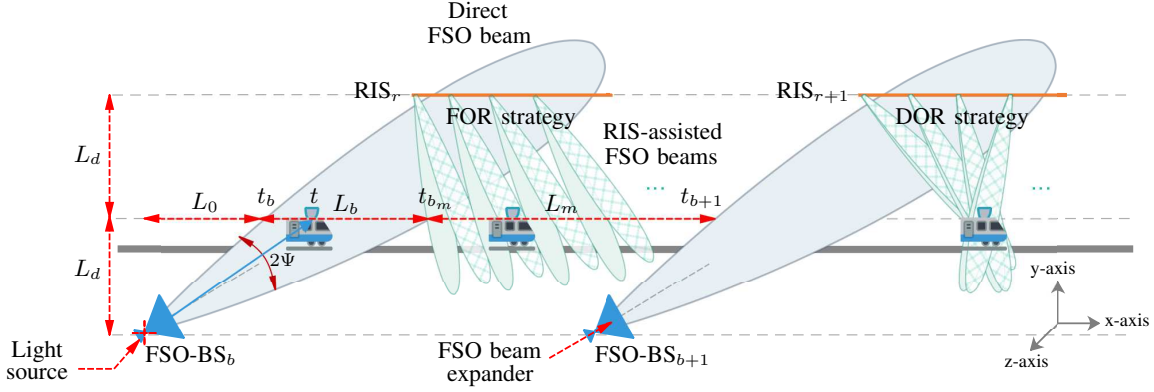}{
    \psfrag{x-axis}{\hspace{0.0cm}\scriptsize x-axis}
    \psfrag{y-axis}{\hspace{0.0cm}\scriptsize y-axis}
    \psfrag{z-axis}{\hspace{0.0cm}\scriptsize z-axis}
    \psfrag{RISr}{\hspace{4.15cm}\small $\text{RIS}_r$}
    \psfrag{RISr+1}{\hspace{-3.9cm}\small $\text{RIS}_{r+1}$}
    \psfrag{FSO-BSb}{\hspace{-0.22cm}\small $\text{FSO-BS}_{b}$}
    \psfrag{FSO-BSb+1}{\hspace{-0.22cm}\small $\text{FSO-BS}_{b+1}$}
    \psfrag{PD}{\hspace{-0.2cm}\small $\text{PD}_\text{BS}$}
    \psfrag{Incident}{\hspace{0.1cm}\small Direct}
    \psfrag{FSObeam}{\hspace{-0.05cm}\small FSO beam}
    \psfrag{A}{\hspace{-6.45cm}\small RIS-assisted}
    \psfrag{B}{\hspace{-6.42cm}\small FSO beams}
    \psfrag{d2}{\hspace{0.0cm}\small $L_d$}
    \psfrag{d1}{\hspace{2.6cm}\small $L_0$}
    \psfrag{Lb}{\hspace{3.65cm}\small $L_b$}
    \psfrag{Lm}{\hspace{-7.3cm}\small $L_m$}
    \psfrag{Light}{\hspace{0.165cm}\small Light}
    \psfrag{Source}{\hspace{0.1cm}\small source}
    \psfrag{FSO}{\hspace{-0.5cm}\small FSO beam}
    \psfrag{Expander}{\hspace{-0.375cm}\small expander}
    \psfrag{C1}{\hspace{2.7cm}\scriptsize }
    \psfrag{C2}{\hspace{3.01cm}\scriptsize }
    \psfrag{C3}{\hspace{3.16cm}\scriptsize $2\Psi$}
    \psfrag{FOR}{\hspace{-9.1cm}\small FOR strategy}
    \psfrag{DOR}{\hspace{-2.43cm}\small DOR strategy}
    \psfrag{tb}{\hspace{3.3cm}\small $t_b$}
    \psfrag{tb+1}{\hspace{9.5cm}\small $t_{b+1}$}
    \psfrag{t}{\hspace{3.55cm}\small $t$}
    \psfrag{tm}{\hspace{4.7cm}\small $t_{b_m}$}
    }
    \caption{The RIS-assisted FSO access network serving an HST system.}
    \label{Fig:Fig.1}
\end{figure*}
\begin{figure}[h!]
    \centering
    \pstool[scale=1.1]{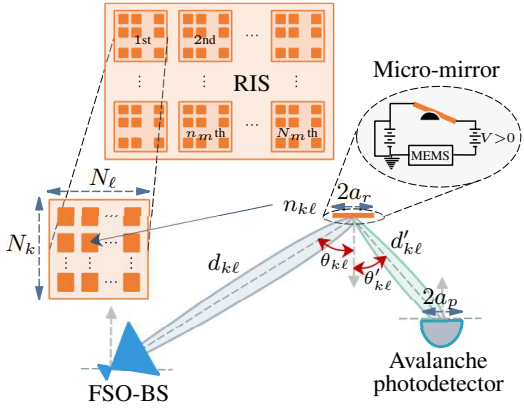}{
    \psfrag{PD1}{\hspace{-0.5cm}\small Avalanche}
    \psfrag{PD2}{\hspace{-0.7cm}\small photodetector}
    \psfrag{FSO-BS}{\hspace{0.23cm}\small FSO-BS}
    \psfrag{Nk}{\hspace{-0.025cm}\small $N_k$}
    \psfrag{Nl}{\hspace{-0.005cm}\small $N_\ell$}
    \psfrag{d1}{\hspace{-3.7cm}\small $d_{k\ell}$}
    \psfrag{d2}{\hspace{-0.68cm}\small $d_{k\ell}^\prime$}
    \psfrag{Elementlk}{\hspace{-1.72cm}\small $n_{k\ell}$}
    \psfrag{2ar}{\hspace{-0.04cm}\small $2a_r$}
    \psfrag{2ap}{\hspace{-0.05cm}\small $2a_p$}
    \psfrag{Tkli}{\hspace{-1.6cm}\scriptsize $\theta_{k\ell}$}
    \psfrag{Tklr}{\hspace{-1.2cm}\scriptsize $\theta_{k\ell}^\prime$}
    \psfrag{RISr}{\hspace{2.55cm}\small $\text{RIS}$}
    \psfrag{1}{\hspace{1.58cm}\tiny $1$st}
    \psfrag{2}{\hspace{1.78cm}\tiny $2$nd}
    \psfrag{Nr}{\hspace{3.21cm}\tiny $N_m$th}
    \psfrag{nr}{\hspace{2.25cm}\tiny $n_m$th}
    \psfrag{V2}{\hspace{-0.085cm}\tiny $V\!\!>\!\!0$}
    \psfrag{ME}{\hspace{-0.13cm}\tiny MEMS}
    \psfrag{FSM}{\hspace{-0.55cm}\small Micro-mirror}
    }
    \caption{The FSO beam's reflection through a RIS-assisted structure.}
    \label{Fig:Fig.2}
\end{figure}
We assume a RIS-assisted FSO access network in which $B$ FSO-BSs and $R$ RISs cooperatively serve an HST system, c.f. Fig.~\ref{Fig:Fig.1}.
In this setup, at each time, the HST is served via either a \emph{direct} FSO beam transmitted from an FSO-BS or \emph{RIS-assisted} FSO beams reflected from a RIS. Each train car is equipped with an avalanche photodetector placed on its roof, which delivers the gathered data from outside to an intra-train network. In Fig.~\ref{Fig:Fig.1}, $L_b$, for $b\!=\!1,2,...,B$, and $L_m$ denote the lengths of the regions covered by direct and RIS-assisted FSO beams, respectively. In addition, $L_d$ and $L_0$ indicate the vertical and horizontal distances of every FSO-BS from the track, respectively. Furthermore, the FSO beam expander at each FSO-BS adjusts its divergence angle to cover the whole area of the paired RIS. Each RIS consists of $N_m$ cells, as shown in Fig.~\ref{Fig:Fig.2}, where each cell provides a concentrated beam by its $N_k\!\times\!N_\ell$ micro-mirror elements. Given that, we consider two coverage strategies based on the RIS-assisted FSO beams; fixed-oriented reflection (FOR) and dynamic-oriented reflection (DOR). In the FOR scenario, the coverage area $L_m$ is divided into $N_m$ one-dimensional grids, where each one is independently covered by a single RIS cell's beam under a fixed orientation. However, in the DOR case, all cells dynamically point their beams to the HST's detector and all them serve the HST simultaneously. In this case, train tracking is required to contrast probable FSO misalignments. As shown in Fig.~\ref{Fig:Fig.1}, $\text{RIS}_{r}$ paired with $\text{FSO-BS}_b$ and $\text{RIS}_{r+1}$ paired with $\text{FSO-BS}_{b+1}$, for $r\!=\!1,2,...,R$, where $\text{RIS}_{r}$ and $\text{RIS}_{r+1}$ work based on the FOR and DOR strategies, respectively.

Before analyzing the network's performance, in what follows, we model the RIS-assisted and direct FSO channels, then we present the received signals at the HST's detector.

\subsection{RIS-Assisted FSO Channel Model}
An end-to-end FSO channel between any FSO-BS and the HST's detector, which is assisted via the $n_{k\ell}$th element of the RIS, as depicted in Fig.~\ref{Fig:Fig.2}, is modeled by $h_{k\ell}\!=\!h_{p,k\ell} h_{t,k\ell} h_{g,k\ell}$. Here, $h_{p,k\ell}$, $h_{t,k\ell}$, and $h_{g,k\ell}$ denote the path-loss, turbulence-induced fading, and the sum of geometric and misalignment losses (GML), respectively, which are discussed as follows.

\subsubsection{Path-Loss} Under the Beer's law, $h_{p,k\ell}$ is given by
\begin{equation}\label{Eq:Eq1}
    h_{p,k\ell} = 10^{-\gamma d_{e,k\ell}},
\end{equation}
where $\gamma$ and $d_{e,k\ell} \!=\! d_{k\ell} \!+\! d_{k\ell}^\prime$ present the attenuation factor and end-to-end distance, respectively.

\subsubsection{Turbulence-Induced Fading} The distribution of $h_{t,k\ell}$ is log-normal, i.e., $h_{t,k\ell}\!\sim\!\operatorname{LN}(\mu_{k\ell},\sigma_{k\ell}^2)$, or Gamma-Gamma, i.e., $h_{t,k\ell}\!\sim\!\operatorname{GG}(\alpha_{k\ell},\beta_{k\ell})$, under weak or moderate-to-strong turbulence condition, sequentially. Hence, we have \cite{khalighi2014survey,agheli2021uav}
\begin{align}\label{Eq:Eq2} \nonumber
&\!f_{h_{t,k\ell}}(h_{t,k\ell}) \!=\!\\ 
   &\!\begin{cases}
     \!\dfrac{1}{h_{t,k\ell} \sqrt{8 \pi \sigma_{k\ell}^2}} \operatorname{exp} \!\left(\!\!- \dfrac{\left(\ln(h_{t,k\ell}) - 2\mu_{k\ell}\right)^2}{8 \sigma_{k\ell}^2} \!\right)\!\!,~{\!{\text{\small for LN}}},\\
     \!\dfrac{2(\!\alpha_{k\ell}\beta_{k\ell}h_{t,k\ell}\!)^{\!\frac{\alpha_{k\ell}\!+\!\beta_{k\ell}}{2}}}{\Gamma(\alpha_{k\ell})\Gamma(\beta_{k\ell})h_{t,k\ell}} K_{\alpha_{k\ell}-\beta_{k\ell}}\!\Big(\!2 \sqrt{\!\alpha_{k\ell}\beta_{k\ell}h_{t,k\ell}}\Big)\!,~{\!\text{\small for GG}},
   \end{cases} 
\end{align}
where $\mu_{k\ell} \!=\! -\sigma_{k\ell}^2$, $\sigma_{k\ell}^2 \!=\! 0.307 C_n^2 k^{7/6} d_{e,k\ell}^{11/6}$, $C_n^2$ represents the refraction structure's index, $k\!=\!{2 \pi}/\lambda$ is the optical wave number, and $\lambda$ denotes the wavelength. Also, we have \cite{khalighi2014survey}
\begin{subequations}\label{Eq:Eq3}
\begin{align}
    &\alpha_{k\ell} \!=\! \bigg[\! \operatorname{exp} \!\left( \dfrac{1.96\sigma_{k\ell}^2}{(1\!+\!4.44\sigma_{k\ell}^{12/5})^{7/6}}\right) \!-\! 1 \bigg]^{\!-1},\\
    &\beta_{k\ell} \!=\! \bigg[\! \operatorname{exp} \!\left( \dfrac{2.04\sigma_{k\ell}^2}{(1\!+\!2.76\sigma_{k\ell}^{12/5})^{5/6}}\right) \!-\! 1 \bigg]^{\!-1}.
\end{align}
\end{subequations}
\subsubsection{GML} Given the positions and orientations of the assumed FSO-BS, RIS, and HST's detector, the conditional GML is given by the following p.d.f \cite{najafi2019intelligent}
\begin{equation}\label{Eq:Eq4}
    f_{h_{g,k\ell}}(h_{g,k\ell}) \!=\! \dfrac{\sqrt{\varpi_{k\ell}}}{2\sqrt{\pi}h_{0,k\ell}} \!\left[ \ln\!\left(\! \dfrac{h_{0,k\ell}}{h_{g,k\ell}} \!\right) \!\right]^{\!-1/2} \!\!\left(\! \dfrac{h_{g,k\ell}}{h_{0,k\ell}} \!\right)^{\!\!\varpi_{k\ell\!-\!1}},
\end{equation}
if $0\!\leq\!h_{g,k\ell}\!\leq\!h_{0,k\ell}$, where $h_{0,k\ell}\!=\! \operatorname{erf}(v
_{k\ell})$, $v_{k\ell} \!=\! {\sqrt{2} \cos(\theta_{k\ell}^\prime)  }{w^{\!-1}(d_{e,k\ell},\hat{w}_0)} a_p$, and we have
\begin{subequations}
\begin{align}\label{Eq:Eq5a}
    &w(d_{e,k\ell},\hat{w}_0) \!=\! \hat{w}_0 \sqrt{1 \!+\! \left(\! \dfrac{\lambda d_{e,k\ell}}{\pi \hat{w}_0^2} \!\right)^{\!\!2}} ,\\\label{Eq:Eq5b}
    &\varpi_{k\ell} \!=\! \dfrac{\sqrt{\pi}}{8} \dfrac{h_{0,k\ell}w^2(d_{e,k\ell},\hat{w}_0)}{v_{k\ell} \operatorname{exp}(-v_{k\ell}^2) \cos^2(\theta_{k\ell}^\prime) \delta_m^2}, \\ \label{Eq:Eq5c}
    &\delta_{m}^2 \!=\! \dfrac{1}{\cos^2(\theta_{k\ell}^\prime)} \!\left(\!\dfrac{\cos^2(\theta_{k\ell}^\prime)}{\cos^2(\theta_{k\ell})} \delta_{s}^{2}\!+\!\dfrac{\sin^2(\theta_{k\ell}\!+\!\theta_{k\ell}^\prime)}{\cos^2(\theta_{k\ell})} \delta_{r}^{2}\!+\!\delta_{l}^{2}\!\right)\!.
\end{align}
\end{subequations}
In (\ref{Eq:Eq5c}), $\delta_s^2$, $\delta_r^2$, and $\delta_l^2$ denote variances of building sway fluctuations at the FSO-BS, RIS, and detector, respectively~\cite{najafi2019intelligent}.
By using (\ref{Eq:Eq1}), (\ref{Eq:Eq2}), and (\ref{Eq:Eq4}), the p.d.fs of $h_{k\ell}$ for the two turbulence conditions are computed in the subsequent theorems.
\begin{theorem}\label{The:The1} The p.d.f of $h_{k\ell}$ under the weak turbulence is
\begin{align}\label{Eq:Eq6} \nonumber
    f_{h_{k\ell}}(h_{k\ell})&= \dfrac{\sqrt{\varpi_{k\ell}}}{8\sqrt{\pi}} \left(\! \dfrac{h_{k\ell}}{\mathcal{C}_a} \!\right)^{\!\!\varpi_{k\ell}} \!\!\operatorname{erfc} \!\left(\! \dfrac{\ln\!\left(\!\frac{h_{k\ell}}{\mathcal{C}_a} \!\right)\!+\!\mathcal{C}_b}{\sqrt{8 \sigma_{k\ell}^2}} \!\right)\\ \nonumber
    &~~\times\! \Bigg\{\!\! \left[\dfrac{5\!-\!\ln(\mathcal{C}_a)}{h_{k\ell}^{\varpi_{k\ell}}}\!+\!\dfrac{\sqrt{8\sigma_{k\ell}^2}}{\sqrt{\pi}}\right] \\ \nonumber
    &~~\times\! \operatorname{exp}\!\left(\!\dfrac{\big(\!\ln(h_{k\ell})\!+\!\mathcal{C}_b\big)^2 \!-\! \big(\!\ln(h_{k\ell})\!+\!2\sigma_{k\ell}^2\big)^2   }{8\sigma_{k\ell}^2} \!\right)\\
    &~~+\! \big(\!\ln(h_{k\ell})\!+\!\mathcal{C}_b\big) \operatorname{exp}\!\Big(\!2\sigma_{k\ell}^2\varpi_{k\ell}(1\!+\!\varpi_{k\ell})\!\Big) \!\!\Bigg\},
\end{align}
where $\mathcal{C}_a\!=\!h_{0,k\ell}h_{p,k\ell}$, and $\mathcal{C}_b\!=\!2\sigma_{k\ell}^2(1\!+\!2\varpi_{k\ell})$.
\end{theorem}
\begin{proof}
See Appendix \ref{App:App1}.
\end{proof}

\begin{theorem}\label{The:The2} The p.d.f of $h_{k\ell}$ under the moderate-to-strong turbulence becomes
\begin{align}\label{Eq:Eq7} \nonumber
    f_{h_{k\ell}}(h_{k\ell})
    &\!=\! \dfrac{\sqrt{\varpi_{k\ell} {\zeta}}}{2\sqrt{\pi}\Gamma(\alpha_{k\ell})\Gamma(\beta_{k\ell})} \dfrac{1}{h_{k\ell}} \left(\!\dfrac{\alpha_{k\ell}\beta_{k\ell}h_{k\ell}}{\mathcal{C}_a}\!\right)^{\!\!\mathcal{C}_c} \\
    &~~\times\!
    \MeijerG*{3}{0}{1}{3}{1+\varpi_{k\ell}}{1+\mathcal{C}_c, \alpha_{k\ell}, \beta_{k\ell}\!}{\!\frac{\alpha_{k\ell}\beta_{k\ell}h_{k\ell}}{\mathcal{C}_a}},
\end{align}
where $\mathcal{C}_c\!=\!1/{2\zeta}\!+\!\varpi_{k\ell}\!-\!1$.
\end{theorem}
\begin{proof}
See Appendix \ref{App:App2}.
\end{proof}

\subsection{Direct FSO Channel Model}
The direct FSO channel between the $b$th FSO-BS and the HST's detector at time $t_b\!\leq\!t\!\leq\!t_{b_m}$ is modeled as $h_{b}(t)\!=\!h_{p,b}(t) h_{t,b}(t) h_{g,b}(t)$ in which $h_{p,b}(t)\!=\!10^{-\gamma d_b(t)}$, and
\begin{align}\label{Eq:Eq8}
    d_b(t) \!=\! \sqrt{ L_d^2 + (L_0 + (t \!-\! t_b)V_\text{HST})^2},
\end{align}
where $V_\text{HST}$ denotes the HST's speed. For modeling the path-loss and turbulence-induced fading, the same formulas as in (\ref{Eq:Eq2}) and (\ref{Eq:Eq3}) are used, where  $h_{t,k\ell}$, $\mu_{k\ell}$, $\sigma_{k\ell}^2$, $\alpha_{k\ell}$, $\beta_{k\ell}$, $d_{e,k\ell}$ are replaced with $h_{t,b}(t)$, $\mu_{b}$, $\sigma_{b}^2$, $\alpha_{b}$, $\beta_{b}$, $d_{b}(t)$, sequentially. Also, the GML reduces to the misalignment loss modeled as \cite{agheli2021uav}
\begin{equation}\label{Eq:Eq9}
    f_{h_{g,b}(t)}(h_{g,b}(t)) \!=\! \frac{\xi^2}{h_{0,b}^{\xi^2}} \big(h_{g,p}(t)\big)^{\xi^2-1}, \text{~for~} 0\!\leq\! h_{g,b}(t) \!\leq\! h_{0,b},
\end{equation}
where $h_{0,b}$ and $\xi$ denote misalignment parameters. 
Thus, the p.d.f of $h_{b}(t)$ under the weak turbulence is given by \cite{farid2007outage}
\begin{align}\label{Eq:Eq10}
    f_{h_b(t)}(h_b(t)) \!=\! \frac{\xi^2 (h_b(t))^{\xi^2-1}}{2 \mathcal{C}_d^{\xi^2}} \operatorname{erfc}\! \left(\! \frac{\ln\!\left(\!{\frac{h_b(t)}{\mathcal{C}_d}}\!\right) \!+\! \mathcal{C}_e}{\sqrt{8 \sigma_{b}^2}}\!\right) \!\mathcal{C}_f,
\end{align}
where $\sigma_{b}^2 \!=\! 0.307 C_n^2 k^{7/6} (d_{b}(t))^{11/6}$, $\mathcal{C}_d\!=\!h_{0,b}h_{p,b}(t)$, $\mathcal{C}_e \!=\! 2 \sigma_{b}^2\! \left(1+2\xi^2\right)$, and $\mathcal{C}_f \!=\! 2 \sigma_{b}^2 \xi^2 \! \left(1+\xi^2\right)$. The p.d.f of $h_{b}(t)$ under the moderate-to-strong turbulence is derived in the following theorem.
\begin{theorem}\label{The:The3} The p.d.f of $h_{b}(t)$ under the moderate-to-strong turbulence is
\begin{align}\label{Eq:Eq11} \nonumber
    f_{h_{b}(t)}(h_{b}(t)) &\!=\! \dfrac{\xi^2}{\Gamma(\alpha_{b})\Gamma(\beta_{b})} \dfrac{1}{h_b(t)} \left(\!\dfrac{\alpha_{b}\beta_{b}h_b(t)}{\mathcal{C}_d}\!\right)^{\!\!\xi^2}\\ 
    &~~\times\! \MeijerG*{3}{0}{1}{3}{1+\xi^2}{\xi^2, \alpha_{b}, \beta_{b}\!}{\!\frac{\alpha_{b}\beta_{b}h_{b}(t)}{\mathcal{C}_d}}.
\end{align}
\end{theorem}
\begin{proof}
See Appendix \ref{App:App3}.
\end{proof}

\subsection{HST Received Signals}
In this section, we model the FSO received signals at the HST's detector for the FOR and DOR strategies. To this end, we recall $h_{k\ell}$ as $h_{m,k\ell}$ with similar properties and dedicated to the $n_m$ RIS cell, for $n_m\!=\!1,2,...,N_m$. Let us assume an intensity modulation/direct detection (IM/DD) FSO system, where $s_b(t)\!\in\!\{0,\sqrt{P}\}$ denotes the transmitted symbol from the $b$th FSO-BS at time $t$. According to Fig.~\ref{Fig:Fig.1}, the HST's detector receives the signal from the $b$th FSO-BS via direct FSO channels if $t_b\!\leq\!t\!<\!t_{b_m}$ and via RIS-assisted ones if $t_{b_m}\!\leq\!t\!<\!t_{b+1}$. Therefore, the received signal is modeled as\footnote{We assume that time shifts between the transmitted and received direct or RIS-assisted signals are ignitable thanks to short-range but high-rate links.}
\begin{equation}\label{Eq:Eq12}
\begin{aligned}
    &r_b(t) \!=\! \\
    &\begin{cases}
    \big(\eta h_b(t) s_b(t) \!+\! \omega(t) \big) \mathbb{I}(\mathcal{C}_1),\\
    \text{\small FOR:~} \!\!\left\{\!\eta \!\left(\sum\limits_{\ell=1}^{N_\ell}\sum\limits_{k=1}^{N_k} \rho_{k\ell} h_{m,k\ell} \!\right)\! s_b(t) \!+\! \omega^\prime(t)\!\right\}\!
    \mathbb{I}(\mathcal{C}_2),\\
    \text{\small DOR:~} \!\!\left\{\!\eta \!\left(\sum\limits_{n_m=1}^{N_m}\sum\limits_{\ell=1}^{N_\ell}\sum\limits_{k=1}^{N_k} \rho_{k\ell} h_{m,k\ell} \!\right)\! s_b(t) \!+\! \omega^{\prime\prime}(t)\!\right\}\! \mathbb{I}(\mathcal{C}_3),\\
    \end{cases}
\end{aligned}
\end{equation}
where
\begin{align*}
    &\mathcal{C}_1\!: t_b\!\leq\!t\!<\!t_{b_m},\\
    &\mathcal{C}_2\!: t_{b_m}\!\!+\!(n_m\!-\!1)t_m\!\leq\!t\!<\!t_{b_m}\!\!+\!n_mt_m, \forall n_m,\\
    &\mathcal{C}_3\!: t_{b_m}\!\leq\!t\!<\!t_{b+1}.
\end{align*}
Herein, $\eta$ denotes the optical-to-electrical coefficient, $\rho_{k\ell}$ indicates the RIS reflection coefficient, and  $t_m\!=\!L_m/(N_mV_\text{HST})$ indicates the service time span of the $n_m$'th RIS cell. Furthermore,
$\omega(t)\!\sim\!\mathcal{N}(0,\sigma_{\omega}^2)$, $\omega^\prime(t)\!\sim\!\mathcal{N}(0,\sigma_{\omega^\prime}^{2})$, and $\omega^{\prime\prime}(t)\!\sim\!\mathcal{N}(0,\sigma_{\omega^{\prime\prime}}^{2})$ are additive white Gaussian noise terms.

\section{Performance Analyses}\label{Sec:Sec3}
Through this section, the proposed network's average SNR and outage probability are formulated.

\subsection{Average SNR}
If we assume the RIS-assisted channels are mutually uncorrelated, the SNR of the $b$th FSO-BS link is obtained as
\begin{equation}\label{Eq:Eq13}
\begin{aligned}
    \gamma_b(t) \!=\!
    \begin{cases}
    \bar{\gamma}_b (h_b(t))^2 \mathbb{I}(\mathcal{C}_1),\\
    \text{\small FOR:~} \!\left\{\sum\limits_{\ell=1}^{N_\ell}\sum\limits_{k=1}^{N_k} \bar{\gamma}_{k\ell} h_{m,k\ell}^2\!\right\}\! \mathbb{I}(\mathcal{C}_2),\\
    \text{\small DOR:~}  \!\left\{\sum\limits_{n_m=1}^{N_m}\sum\limits_{\ell=1}^{N_\ell}\sum\limits_{k=1}^{N_k} \bar{\gamma}^\prime_{k\ell} h_{m,k\ell}^2\!\right\}\!\mathbb{I}(\mathcal{C}_3),\\
    \end{cases}
\end{aligned}
\end{equation}
wherein we have $\bar{\gamma}_b\!=\!\eta^2 P / \sigma_\omega^2$, $\bar{\gamma}_{k\ell}\!=\!\eta^2 \rho_{k\ell}^2 P / \sigma_{\omega^\prime}^2$, and $\bar{\gamma}^\prime_{k\ell}\!=\!\eta^2 \rho_{k\ell}^2 P / \sigma_{\omega^{\prime\prime}}^2$.
Thus, the average SNR is achieved by replacing $(h_b(t))^2$ and $h_{m,k\ell}^2$ with $\Gamma_b^2(t)\!=\!\mathbb{E}\big\{\!(h_b(t))^2\!\big\}$ and $\Gamma_{m,k\ell}^2\!=\!\mathbb{E}\big\{\!h_{m,k\ell}^2\!\big\}$ in (\ref{Eq:Eq13}), respectively. Under the weak turbulence, by the use of (\ref{Eq:Eq6}) and (\ref{Eq:Eq10}), we have
\begin{subequations}
\begin{align}\nonumber
    &\Gamma_{m,k\ell}^2\!\cong\!
    \dfrac{\sqrt{\varpi_{k\ell}} \,\mathcal{C}_b^3}{8\sqrt{\pi}\,\mathcal{C}_a^{\varpi_{k\ell}}}\Bigg[
    \dfrac{\big({5\!-\!\ln(\mathcal{C}_a)\big)}}{3} 
    \Bigg(\!\dfrac{e^{-\mathcal{C}_g^2} \big({\mathcal{C}_g^2+1}\big)}{\sqrt{\pi}\,\mathcal{C}_g^3} \!-\! \operatorname{erfc}\!\big(\mathcal{C}_g\big) \!\Bigg)
    \\ \label{Eq:Eq14a}
    &~+ \dfrac{\sqrt{8\sigma_{k\ell}^2} \,\mathcal{C}_b^{\varpi_{k\ell}}}{{\sqrt{\pi} \left(\varpi_{k\ell}\!+\!3\right)}} \Bigg(\!\!
    \dfrac{1}{\sqrt{\pi}\,\mathcal{C}_g^{\varpi_{k\ell}+3}}
    \operatorname{\Gamma}\!\left(\!\frac{\varpi_{k\ell}\!+\!4}{2},\mathcal{C}_g^2\right) \!-\! \operatorname{erfc}\!\big(\mathcal{C}_g\big)\!\!\Bigg) \!\Bigg],
    \\ \nonumber
    &\Gamma_b^2(t)\!=\!
    \dfrac{\xi^2 \big(\mathcal{C}_e/\mathcal{C}_h\big)^{\!{\xi}^2+2} \mathcal{C}_f}{2\left({\xi}^2\!+\!2\right)\mathcal{C}_d^{\xi^2}}\\ \label{Eq:Eq14b}
    &~~~~~~~~\times\!\Bigg[ 
    \dfrac{1}{\sqrt{\pi}}
    \operatorname{\Gamma}\!\left(\!\frac{{\xi}^2\!+\!3}{2},\mathcal{C}_h^2\right)
    \!-
    \mathcal{C}_h^{\xi^2+2}
    \operatorname{erfc}\!\big(\mathcal{C}_h\big) \!\Bigg],
\end{align}
\end{subequations}
where $\mathcal{C}_g\!=\!\frac{\mathcal{C}_b}{\mathcal{C}_a\sqrt{8 \sigma_{k\ell}^2}}$, and $\mathcal{C}_h\!=\!\frac{\mathcal{C}_e}{\mathcal{C}_d\sqrt{8 \sigma_{b}^2}}$.
However, under the moderate-to-strong turbulence, by using (\ref{Eq:Eq7}) and (\ref{Eq:Eq11}), we have
\begin{subequations}
\begin{align}\nonumber
    &\Gamma_{m,k\ell}^2\!=\! \dfrac{\sqrt{\varpi_{k\ell} {\zeta}}}{2\sqrt{\pi}\Gamma(\alpha_{k\ell})\Gamma(\beta_{k\ell})}  \left(\!\dfrac{\alpha_{k\ell}\beta_{k\ell}}{\mathcal{C}_a}\!\right)^{\!\!\mathcal{C}_c} \\ \label{Eq:Eq15a}
    &~~~~~~~~\times\!
    \MeijerG*{4}{1}{2}{4}{-1-\mathcal{C}_c,1+\varpi_{k\ell}}{-2-\mathcal{C}_c,1+\mathcal{C}_c, \alpha_{k\ell}, \beta_{k\ell}\!}{\!\frac{\alpha_{k\ell}\beta_{k\ell}}{\mathcal{C}_a}}, \\ \nonumber
    &\Gamma_b^2(t)\!=\!\dfrac{\xi^2}{\Gamma(\alpha_{b})\Gamma(\beta_{b})} \left(\!\dfrac{\alpha_{b}\beta_{b}}{\mathcal{C}_d}\!\right)^{\!\!\xi^2}\\ \label{Eq:Eq15b}
    &~~~~~~~~\times\! \MeijerG*{4}{1}{2}{4}{-1-\xi^2,1+\xi^2}{-2-\xi^2,\xi^2, \alpha_{b}, \beta_{b}\!}{\!\frac{\alpha_{b}\beta_{b}}{\mathcal{C}_d}}.
\end{align}
\end{subequations}

\subsection{Outage Probability}
The QoS is ensured by keeping $\gamma_b(t)$ above a given threshold $\gamma_{th}$ at time $t$. Therefore, the network's outage probability is defined as \cite{agheli2021uav}
\begin{align}
    P_{out}(t) \triangleq \operatorname{Pr}\! \big\{\gamma_{b}(t) \!\leq\! \gamma_{th} \big\} = F_{\gamma_{b}(t)}(\gamma_{b}(t)\!=\!\gamma_{th}),
\end{align}
where $F_{\gamma_{b}(t)}(\gamma_{b}(t))$ shows the cumulative distribution function (CDF) of $\gamma_{b}(t)$, which is derived in the following theorems.

\begin{theorem}\label{The:The4} The CDF of $\gamma_{b}(t)$ under the weak turbulence is obtained for the direct and RIS-assisted FSO channels, as what follows.
\begin{itemize}
    \item Direct FSO channels;
    \begin{align}\label{Eq:Eq17} \nonumber
    &F_{\gamma_{b}(t)}(\gamma_{b}(t)) \!=\! \frac{ \big(\mathcal{C}_e/\mathcal{C}_h\big)^{\!\frac{{\xi}^2}{2}} \mathcal{C}_f}{4\, \mathcal{C}_d^{\xi^2}} \\ \nonumber
    &~~~~~~~~\times\! \Bigg\{\!
     \hat{\mathcal{C}}_h^{\frac{{\xi}^2}{2}}\! \!\left[1\!+\!\frac{\gamma_b(t)}{\mathcal{C}_e}\!\right]^{\!\frac{{\xi}^2}{2}}\! \operatorname{erfc}\!\left(\!\hat{\mathcal{C}}_h\!\left[1\!+\!\frac{\gamma_b(t)}{\mathcal{C}_e}\!\right]\!\right) \\
    &~~~~~~~~-\! \frac{1}{\sqrt{{\pi}}} \operatorname{\Gamma}\!\left(\!\frac{{\xi}^2\!+\!2}{4},\hat{\mathcal{C}}_h^2\!\left[1\!+\!\frac{\gamma_b(t)}{\mathcal{C}_e}\!\right]^{\!2}\right) \!\!\!\Bigg\}
    \mathbb{I}(\mathcal{C}_1),
    \end{align}
    \item RIS-assisted FSO channels and the FOR scenario;
    \begin{align}\label{Eq:Eq18} \nonumber
    &F_{\gamma_{b}(t)}(\gamma_{b}(t)) \!=\! \prod\limits_{\ell=1}^{N_\ell}\prod\limits_{k=1}^{N_k} \dfrac{\sqrt{8\sigma_{k\ell}^2} \,\big(\mathcal{C}_b/\mathcal{C}_g\big)^{\!\frac{\varpi_{k\ell}}{2}}}{8{\pi} \sqrt{\varpi_{k\ell}} \,\mathcal{C}_a^{\varpi_{k\ell}}} \\ \nonumber
    &~~~~~~~~\times\! \Bigg\{\!
     \hat{\mathcal{C}}_g^{\frac{\varpi_{k\ell}}{2}}\! \!\left[1\!+\!\frac{\gamma_b(t)}{\mathcal{C}_b}\!\right]^{\!\frac{\varpi_{k\ell}}{2}}\! \operatorname{erfc}\!\left(\!\hat{\mathcal{C}}_g\!\left[1\!+\!\frac{\gamma_b(t)}{\mathcal{C}_b}\!\right]\!\right) \\ 
    &~~~~~~~~-\! \frac{1}{\sqrt{{\pi}}} \operatorname{\Gamma}\!\left(\!\frac{\varpi_{k\ell}\!+\!2}{4},\hat{\mathcal{C}}_g^2\!\left[1\!+\!\frac{\gamma_b(t)}{\mathcal{C}_b}\!\right]^{\!2}\right) \!\!\!\Bigg\} \mathbb{I}(\mathcal{C}_2),
    \end{align}
    \item RIS-assisted FSO channels and the DOR scenario;
    \begin{align}\label{Eq:Eq19} \nonumber
    &F_{\gamma_{b}(t)}(\gamma_{b}(t)) \!=\! \!\! \prod\limits_{n_m=1}^{N_m}
    \prod\limits_{\ell=1}^{N_\ell}\prod\limits_{k=1}^{N_k} \dfrac{\sqrt{8\sigma_{k\ell}^2} \,\big(\mathcal{C}_b/\mathcal{C}_g\big)^{\!\frac{\varpi_{k\ell}}{2}}}{8{\pi} \sqrt{\varpi_{k\ell}} \,\mathcal{C}_a^{\varpi_{k\ell}}} \\\nonumber
    &~~~~~~~~\times\! \Bigg\{\!
     \doublehat{\mathcal{C}}_g^{\frac{\varpi_{k\ell}}{2}}\! \!\left[1\!+\!\frac{\gamma_b(t)}{\mathcal{C}_b}\!\right]^{\!\frac{\varpi_{k\ell}}{2}}\! \operatorname{erfc}\!\left(\!\doublehat{\mathcal{C}}_g\!\left[1\!+\!\frac{\gamma_b(t)}{\mathcal{C}_b}\!\right]\!\right) \\ 
    &~~~~~~~~-\! \frac{1}{\sqrt{{\pi}}} \operatorname{\Gamma}\!\left(\!\frac{\varpi_{k\ell}\!+\!2}{4},\doublehat{\mathcal{C}}_g^2\!\left[1\!+\!\frac{\gamma_b(t)}{\mathcal{C}_b}\!\right]^{\!2}\right) \!\!\!\Bigg\} \mathbb{I}(\mathcal{C}_3),
    \end{align}
\end{itemize}
where $\hat{\mathcal{C}}_h\!=\!{\mathcal{C}_h}/{\sqrt{\bar{\gamma}_b}}$, $\hat{\mathcal{C}}_g\!=\!{\mathcal{C}_g}/{\sqrt{\bar{\gamma}_{k\ell}}}$, and $\doublehat{\mathcal{C}}_g\!=\!{\mathcal{C}_g}/{\sqrt{\bar{\gamma}_{k\ell}^\prime}}\,$.
\end{theorem}
\begin{proof}
See Appendix \ref{App:App4}.
\end{proof}

\begin{theorem}\label{The:The5} The CDF of $\gamma_{b}(t)$ under the moderate-to-strong turbulence is obtained for the direct and RIS-assisted FSO channels, as what follows.
\begin{itemize}
    \item Direct FSO channels;
    \begin{align}\label{Eq:Eq20} \nonumber
    &F_{\gamma_{b}(t)}(\gamma_{b}(t)) \!=\! \dfrac{\xi^2}{\bar{\gamma}_b^\frac{\xi^2}{2} \Gamma(\alpha_{b})\Gamma(\beta_{b})} \!\left(\!\dfrac{\alpha_{b}\beta_{b}}{\mathcal{C}_d}\!\right)^{\!\!\xi^2}\\ 
    &~~~~~~~\times\! \MeijerG*{4}{1}{2}{4}{1-\xi^2,1+\xi^2}{-\xi^2,\xi^2, \alpha_{b}, \beta_{b}\!}{\!\frac{\alpha_{b}\beta_{b}\sqrt{\gamma_b(t)}}{\sqrt{\bar{\gamma}_{b}}\,\mathcal{C}_d}} \mathbb{I}(\mathcal{C}_1),
    \end{align}
    \item RIS-assisted FSO channels and the FOR scenario;
    \begin{align}\label{Eq:Eq21} \nonumber
    &F_{\gamma_{b}(t)}(\gamma_{b}(t)) \!=\! \prod\limits_{\ell=1}^{N_\ell}\prod\limits_{k=1}^{N_k} \dfrac{\sqrt{\varpi_{k\ell} {\zeta}}}{\sqrt{\pi} \bar{\gamma}_{k\ell}^{\frac{\mathcal{C}_c}{2}} \Gamma(\alpha_{k\ell})\Gamma(\beta_{k\ell})}  \!\left(\!\dfrac{\alpha_{k\ell}\beta_{k\ell}}{\mathcal{C}_a}\!\right)^{\!\!\mathcal{C}_c} \\
    &~\times\!
    \MeijerG*{4}{1}{2}{4}{1-\mathcal{C}_c,1+\varpi_{k\ell}}{-\mathcal{C}_c,1+\mathcal{C}_c, \alpha_{k\ell}, \beta_{k\ell}\!}{\!\frac{\alpha_{k\ell}\beta_{k\ell}\sqrt{\gamma_b(t)}}{\sqrt{\bar{\gamma}_{k\ell}}\,\mathcal{C}_a}} \mathbb{I}(\mathcal{C}_2),
    \end{align}
    \item RIS-assisted FSO channels and the DOR scenario;
    \begin{align}\label{Eq:Eq22} \nonumber
    &\!\!\!\!\!\!\!\!\!\!F_{\gamma_{b}(t)}(\gamma_{b}(t)) \!=\!\!\! \prod\limits_{n_m=1}^{N_m} \prod\limits_{\ell=1}^{N_\ell}\prod\limits_{k=1}^{N_k} \dfrac{\sqrt{\varpi_{k\ell} {\zeta}}}{\sqrt{\pi} \bar{\gamma}_{k\ell}^{\prime \frac{\mathcal{C}_c}{2}} \Gamma(\alpha_{k\ell})\Gamma(\beta_{k\ell})} \!\left(\!\dfrac{\alpha_{k\ell}\beta_{k\ell}}{\mathcal{C}_a}\!\right)^{\!\!\mathcal{C}_c}\\
    &~\times 
    \!\MeijerG*{4}{1}{2}{4}{1-\mathcal{C}_c,1+\varpi_{k\ell}}{-\mathcal{C}_c,1+\mathcal{C}_c, \alpha_{k\ell}, \beta_{k\ell}\!}{\!\frac{\alpha_{k\ell}\beta_{k\ell}\sqrt{\gamma_b(t)}}{\sqrt{\bar{\gamma}_{k\ell}^\prime}\,\mathcal{C}_a}} \mathbb{I}(\mathcal{C}_3).
    \end{align}
\end{itemize}
\end{theorem}
\begin{proof}
See Appendix \ref{App:App5}.
\end{proof}


\section{Numerical Results and Discussions}\label{Sec:Sec4}
In this section, we firstly analyze the network's performance under the assumed serving scenarios and compare the proposed RIS-assisted systems with the existing relay-based ones from the average SNR and outage probability perspectives. The system model parameters used for numerical results are summarized in Table~\ref{Tab:Tab.1}. 
Then, we provide a system design framework with an illustrative sample for practical use cases.
\noindent
\renewcommand{\arraystretch}{1}
\begin{table}[t!]
\begin{center}
\caption{System parameters for numerical results.}\label{Tab:Tab.1}
\begin{tabular}{ | l | c | l |}
\hline
\multicolumn{1}{|c|}{\small \text{Parameter}} & \small \text{Symbol} & \multicolumn{1}{c|}{\small \text{Value}}\\
\hline
\hline
\!\! \footnotesize Vertical distance along the track & \footnotesize  $L_d$ &  \footnotesize $2.5\,[\text{m}]$\\
\hline
\!\! \footnotesize Half divergence angle & \footnotesize  $\Psi$ &  \footnotesize $3.5\degree$\\
\hline
\!\! \footnotesize RIS setup & \footnotesize  $\!\!(N_k,N_\ell,N_m)\!\!$ &  \footnotesize $(10,10,25)$\\
\hline
\!\! \footnotesize RIS reflection coefficient (passive)& \footnotesize  $\rho_{k\ell}$ &  \footnotesize $0.95$\\
\hline
\!\! \footnotesize Attenuation factor & \footnotesize  $\gamma$ &  \footnotesize $0.44\,[\text{dB/km}]$\\
\hline
\!\! \footnotesize Refraction structure index & \footnotesize  $C_n^2$ &  \footnotesize $10^{-15}\, [\text{m}^{-2/3}]$\\
\hline
\!\! \footnotesize Wavelength & \footnotesize  $\lambda$ &  \footnotesize $850\, [\text{nm}]$\\
\hline
\!\! \footnotesize Source beam waist & \footnotesize  $\hat{w}_0$ &  \footnotesize $\lambda/2 \pi \Psi$\\
\hline
\!\! \footnotesize Building sway fluctuations & \footnotesize $ \delta_s,\delta_r,\delta_l$ &  \footnotesize $5$ \\
\hline
\!\! \footnotesize Detector diameters & \footnotesize $( \alpha_r,\alpha_p)$ &  \footnotesize $(2.5,20) [\text{cm}]$\\
\hline
\!\! \footnotesize Approximation factor & \footnotesize $\zeta$ &  \footnotesize $100$ \\
\hline
\!\! \footnotesize Misalignment parameters & \footnotesize $(h_{0,b},\xi)$ &  \footnotesize $(0.0764,2.35)$ \\
\hline
\footnotesize Optical transmission power & \footnotesize $P$ &  \footnotesize $40\,[\text{mW}]$\\
\hline
\!\! \footnotesize Optical-to-electrical coefficient & \footnotesize $\eta$ &  \footnotesize $0.5$\\
\hline
\!\! \footnotesize Speed of the HST & \footnotesize $V_\text{HST}$ &  \footnotesize $300\,[\text{km/h}]$\\
\hline
\!\! \footnotesize Additive noise standard deviations & \footnotesize $\!\!\sigma_{\omega},\sigma_{\omega^\prime},\sigma_{\omega^{\prime\prime}}\!\!$ &  \footnotesize $10^{-7}\,[\text{A/Hz}]$\\
\hline
\end{tabular}
\medskip
\end{center}
\end{table}
\begin{figure}[t!]
\centering
\pstool[scale=0.55]{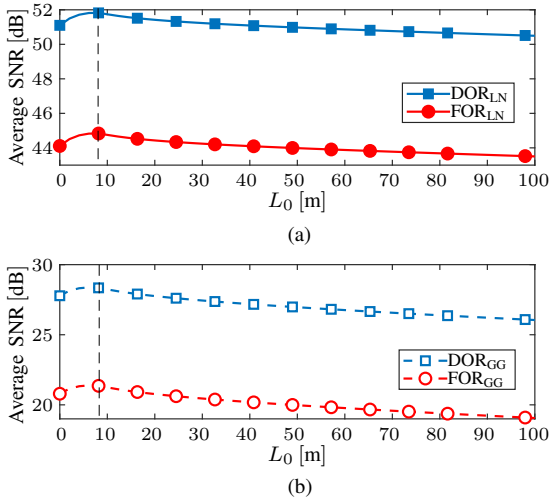}{
\psfrag{Average-SNR-[dB]}{\hspace{-0.265cm}\footnotesize Average SNR $\![\text{dB}]$}
\psfrag{L0-[m]}{\hspace{-0.1cm}\footnotesize $L_0\,[\text{m}]$}
\psfrag{(a)}{\hspace{0.0cm}\footnotesize (a)}
\psfrag{(b)}{\hspace{0.0cm}\footnotesize (b)}
\psfrag{FOR-LN1}{\hspace{-0.05cm}\scriptsize $\text{FOR}_\text{LN}$}
\psfrag{FOR-GG1}{\hspace{-0.05cm}\scriptsize $\text{FOR}_\text{GG}$}
\psfrag{DOR-LN}{\hspace{-0.05cm}\scriptsize $\text{DOR}_\text{LN}$}
\psfrag{DOR-GG}{\hspace{-0.05cm}\scriptsize $\text{DOR}_\text{GG}$}
\psfrag{0}{\hspace{-0.0cm}\scriptsize $0$}
\psfrag{10}{\hspace{-0.08cm}\scriptsize $10$}
\psfrag{20}{\hspace{-0.08cm}\scriptsize $20$}
\psfrag{25}{\hspace{-0.08cm}\scriptsize $25$}
\psfrag{30}{\hspace{-0.08cm}\scriptsize $30$}
\psfrag{40}{\hspace{-0.08cm}\scriptsize $40$}
\psfrag{44}{\hspace{-0.08cm}\scriptsize $44$}
\psfrag{46}{\hspace{-0.08cm}\scriptsize $46$}
\psfrag{48}{\hspace{-0.08cm}\scriptsize $48$}
\psfrag{50}{\hspace{-0.08cm}\scriptsize $50$}
\psfrag{52}{\hspace{-0.08cm}\scriptsize $52$}
\psfrag{60}{\hspace{-0.08cm}\scriptsize $60$}
\psfrag{70}{\hspace{-0.08cm}\scriptsize $70$}
\psfrag{80}{\hspace{-0.08cm}\scriptsize $80$}
\psfrag{90}{\hspace{-0.08cm}\scriptsize $90$}
\psfrag{100}{\hspace{-0.08cm}\scriptsize $100$}
}
\vspace{-0.2cm}
\caption{The average SNR over the assumed $1\,[\text{km}]$ distance between two adjacent FSO-BSs versus $L_0$ for (a) weak turbulence and (b) moderate-to-strong turbulence conditions.}
\label{Fig:Fig.4}
\end{figure}

\subsection{Results and Discussions}
In Fig.~\ref{Fig:Fig.4}, we show that there exists an optimal value of $L_0$ by plotting the average SNR over an assumed $1\,[\text{km}]$ distance between two adjacent FSO-BSs. For the considered conditions, we observe that $L_0\!\simeq\!8$ is optimal for all scenarios. 

Fig.~\ref{Fig:Fig.5} depicts the network's outage probability versus the average SNR for the two considered turbulence conditions. In this figure, we compare the performances of the proposed systems based on the FOR and DOR scenarios with traditional relay systems. All results are compared with $500$-iteration Monte-Carlo simulations. It is shown that the DOR scenario offers the highest performance compared to the others, where the FOR scenario decreases the network's outage probability with respect to the relay system.
The reason is that, in the DOR scenario, the incident FSO beams at all elements of every RIS are reflected in a concentrated direction, which boosts the received signal power at the HST's detector. However, in the FOR scenario, the incident FSO beams at each RIS cell are reflected in a dedicated direction, independent of the direction of other cells.
Conventionally, we assume that only a limited number of RIS elements reflect the incident FSO beams in one direction in the case of the relay system. If we assume to have a generic passive FSO detector and transmitter instead of each RIS, the number of RIS elements required for simulating the detector's aperture is $(\alpha_p/\alpha_r)^2$. Finally, it is verified that all systems have lower outage probabilities under weak turbulence conditions as opposed to the moderate-to-strong ones.
As clearly observed from Fig.~\ref{Fig:Fig.5}, our derived closed-form formulas provide an identical match to simulation results. This conclusion is extendable to the rest of the figures.
 
\begin{figure}[t!]
\centering
\pstool[scale=0.55]{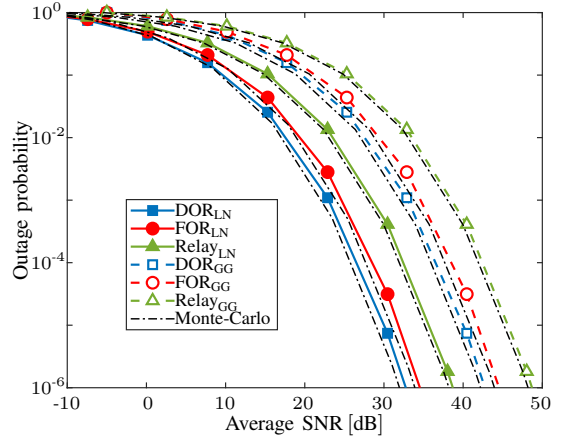}{
\psfrag{SNR-[dB]}{\hspace{-0.62cm}\footnotesize Average SNR $\![\text{dB}]$}
\psfrag{Outage-probability}{\hspace{-0.2cm}\footnotesize Outage probability}
\psfrag{FOR-LN}{\hspace{-0.05cm}\scriptsize $\text{FOR}_\text{LN}$}
\psfrag{FOR-GG}{\hspace{-0.05cm}\scriptsize $\text{FOR}_\text{GG}$}
\psfrag{DOR-LN}{\hspace{-0.05cm}\scriptsize $\text{DOR}_\text{LN}$}
\psfrag{DOR-GG}{\hspace{-0.05cm}\scriptsize $\text{DOR}_\text{GG}$}
\psfrag{Relay-LN}{\hspace{-0.05cm}\scriptsize $\text{Relay}_\text{LN}$}
\psfrag{Relay-GG}{\hspace{-0.05cm}\scriptsize $\text{Relay}_\text{GG}$}
\psfrag{Monte-Carlo123}{\hspace{-0.05cm}\scriptsize Monte-Carlo}
\psfrag{-10}{\hspace{-0.05cm}\scriptsize -$10$}
\psfrag{0}{\hspace{-0.00cm}\scriptsize $0$}
\psfrag{10}{\hspace{-0.06cm}\scriptsize $10$}
\psfrag{20}{\hspace{-0.06cm}\scriptsize $20$}
\psfrag{25}{\hspace{-0.08cm}\scriptsize $25$}
\psfrag{30}{\hspace{-0.06cm}\scriptsize $30$}
\psfrag{40}{\hspace{-0.06cm}\scriptsize $40$}
\psfrag{50}{\hspace{-0.06cm}\scriptsize $50$}
\psfrag{0.4}{\hspace{-0.14cm}\scriptsize $10^{\text{-}6}$}
\psfrag{0.6}{\hspace{-0.14cm}\scriptsize $10^{\text{-}4}$}
\psfrag{0.8}{\hspace{-0.14cm}\scriptsize $10^{\text{-}2}$}
\psfrag{1.0}{\hspace{-0.14cm}\scriptsize $10^0$}
}
\caption{The network's outage probability versus the average SNR for different coverage scenarios and turbulence conditions.}
\label{Fig:Fig.5}
\end{figure}

Fig.~\ref{Fig:Fig.6} illustrates the diameters of the coverage areas served by the proposed access systems under the two considered turbulence conditions. In this regard, the average SNR at the HST's detector is plotted versus $L_b\!+\!L_m$ for the sample $b$th FSO-BS. It is shown that the DOR scenario extends the coverage area with respect to both FOR and traditional relay systems. Besides, the FOR scenario offers a wider coverage area compared to the relay system. In detail, for the average SNR of $30\,[\text{dB}]$, the diameter of the coverage area offered by the DOR scenario is $60\%$ and $304\%$, under the weak turbulence, and $145\%$ and $277\%$, under the moderate-to-strong turbulence, larger than those of the FOR and relay systems, respectively. By increasing the coverage area served by each FSO-BS and its paired RIS, the numbers of required FSO-BSs and handovers decrease for a fixed railway distance. The dotted area in Fig.~\ref{Fig:Fig.6} depicts the region directly covered by the $b$th FSO-BS with the diameter of $L_b$. It is inferred that, contrary to the FOR and relay systems, the DOR scenario results in higher average SNR at the end of $L_b$ and beginning of $L_m$, thanks to concentrating all dispersed beams on one point. 
The results of Fig.~\ref{Fig:Fig.6} verify that the assumed distance between two adjacent FSO-BSs for plotting Fig.~\ref{Fig:Fig.4} is feasible.
\begin{figure}[t!]
\centering
\pstool[scale=0.55]{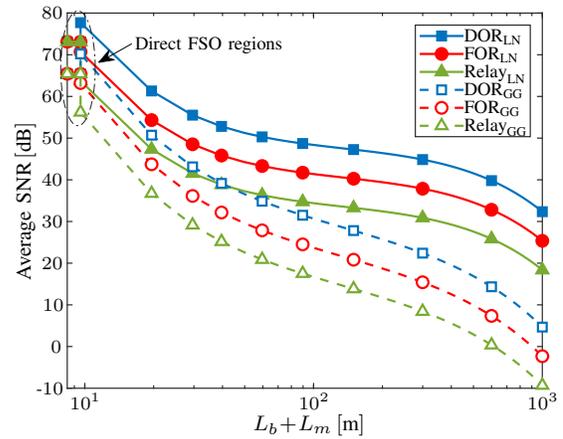}{
\psfrag{SNR-[dB]}{\hspace{-0.62cm}\footnotesize Average SNR $\![\text{dB}]$}
\psfrag{Distance-[m]}{\hspace{-0.1cm}\footnotesize $L_b\!+\!L_m \,[\text{m}]$}
\psfrag{DIrect-FSO-connections}{\hspace{-0.05cm}\scriptsize Direct FSO regions}
\psfrag{FOR-LN}{\hspace{-0.05cm}\scriptsize $\text{FOR}_\text{LN}$}
\psfrag{FOR-GG}{\hspace{-0.05cm}\scriptsize $\text{FOR}_\text{GG}$}
\psfrag{DOR-LN}{\hspace{-0.05cm}\scriptsize $\text{DOR}_\text{LN}$}
\psfrag{DOR-GG}{\hspace{-0.05cm}\scriptsize $\text{DOR}_\text{GG}$}
\psfrag{Relay-LN}{\hspace{-0.05cm}\scriptsize $\text{Relay}_\text{LN}$}
\psfrag{Relay-GG}{\hspace{-0.05cm}\scriptsize $\text{Relay}_\text{GG}$}
\psfrag{-10}{\hspace{-0.05cm}\scriptsize -$10$}
\psfrag{0}{\hspace{-0.03cm}\scriptsize $0$}
\psfrag{10}{\hspace{-0.05cm}\scriptsize $10$}
\psfrag{20}{\hspace{-0.05cm}\scriptsize $20$}
\psfrag{30}{\hspace{-0.05cm}\scriptsize $30$}
\psfrag{40}{\hspace{-0.05cm}\scriptsize $40$}
\psfrag{50}{\hspace{-0.05cm}\scriptsize $50$}
\psfrag{60}{\hspace{-0.05cm}\scriptsize $60$}
\psfrag{70}{\hspace{-0.05cm}\scriptsize $70$}
\psfrag{80}{\hspace{-0.05cm}\scriptsize $80$}
\psfrag{101}{\hspace{-0.05cm}\scriptsize $10^{1}$}
\psfrag{102}{\hspace{-0.05cm}\scriptsize $10^{2}$}
\psfrag{103}{\hspace{-0.05cm}\scriptsize $10^{3}$}
}
\caption{The average SNR versus the coverage diameter for different coverage scenarios and turbulence conditions.}
\label{Fig:Fig.6}
\end{figure}

Fig.~\ref{Fig:Fig.7} presents the network's achievable data rates versus the diameters of the coverage areas. This figure's attribute is similar to that of Fig.~\ref{Fig:Fig.6} under different serving scenarios and turbulence conditions. Concretely, the DOR scenario averagely offers $16\%$ and $45\%$ higher data rates with respect to the FOR and relay systems, respectively, under weak turbulence conditions. Likewise, under moderate-to-strong conditions, the DOR scenario sequentially provides $26\%$ and $63\%$ higher data rates compared to the FOR and relay systems. Meanwhile, the coverage analyses for Fig.~\ref{Fig:Fig.6} can be easily applied for Fig.~\ref{Fig:Fig.7}.
\begin{figure}[t!]
\centering
\pstool[scale=0.55]{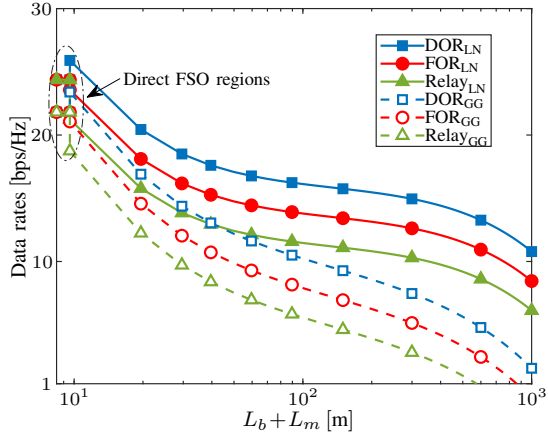}{
\psfrag{Data-rates-[bps/Hz]}{\hspace{-0.18cm}\footnotesize Data rates $\![\text{bps/Hz}]$}
\psfrag{Distance-[m]}{\hspace{-0.1cm}\footnotesize $L_b\!+\!L_m \,[\text{m}]$}
\psfrag{DIrect-FSO-connections}{\hspace{-0.05cm}\scriptsize Direct FSO regions}
\psfrag{FOR-LN}{\hspace{-0.05cm}\scriptsize $\text{FOR}_\text{LN}$}
\psfrag{FOR-GG}{\hspace{-0.05cm}\scriptsize $\text{FOR}_\text{GG}$}
\psfrag{DOR-LN}{\hspace{-0.05cm}\scriptsize $\text{DOR}_\text{LN}$}
\psfrag{DOR-GG}{\hspace{-0.05cm}\scriptsize $\text{DOR}_\text{GG}$}
\psfrag{Relay-LN}{\hspace{-0.05cm}\scriptsize $\text{Relay}_\text{LN}$}
\psfrag{Relay-GG}{\hspace{-0.05cm}\scriptsize $\text{Relay}_\text{GG}$}
\psfrag{1}{\hspace{-0.03cm}\scriptsize $1$}
\psfrag{10}{\hspace{-0.05cm}\scriptsize $10$}
\psfrag{20}{\hspace{-0.05cm}\scriptsize $20$}
\psfrag{30}{\hspace{-0.05cm}\scriptsize $30$}
\psfrag{101}{\hspace{-0.05cm}\scriptsize $10^{1}$}
\psfrag{102}{\hspace{-0.05cm}\scriptsize $10^{2}$}
\psfrag{103}{\hspace{-0.05cm}\scriptsize $10^{3}$}
}
\caption{The network's achievable data rates versus the coverage diameter for different coverage scenarios and turbulence conditions.}
\label{Fig:Fig.7}
\end{figure}

\begin{figure}[t!]
\centering
\subfloat[]{
\pstool[scale=0.55]{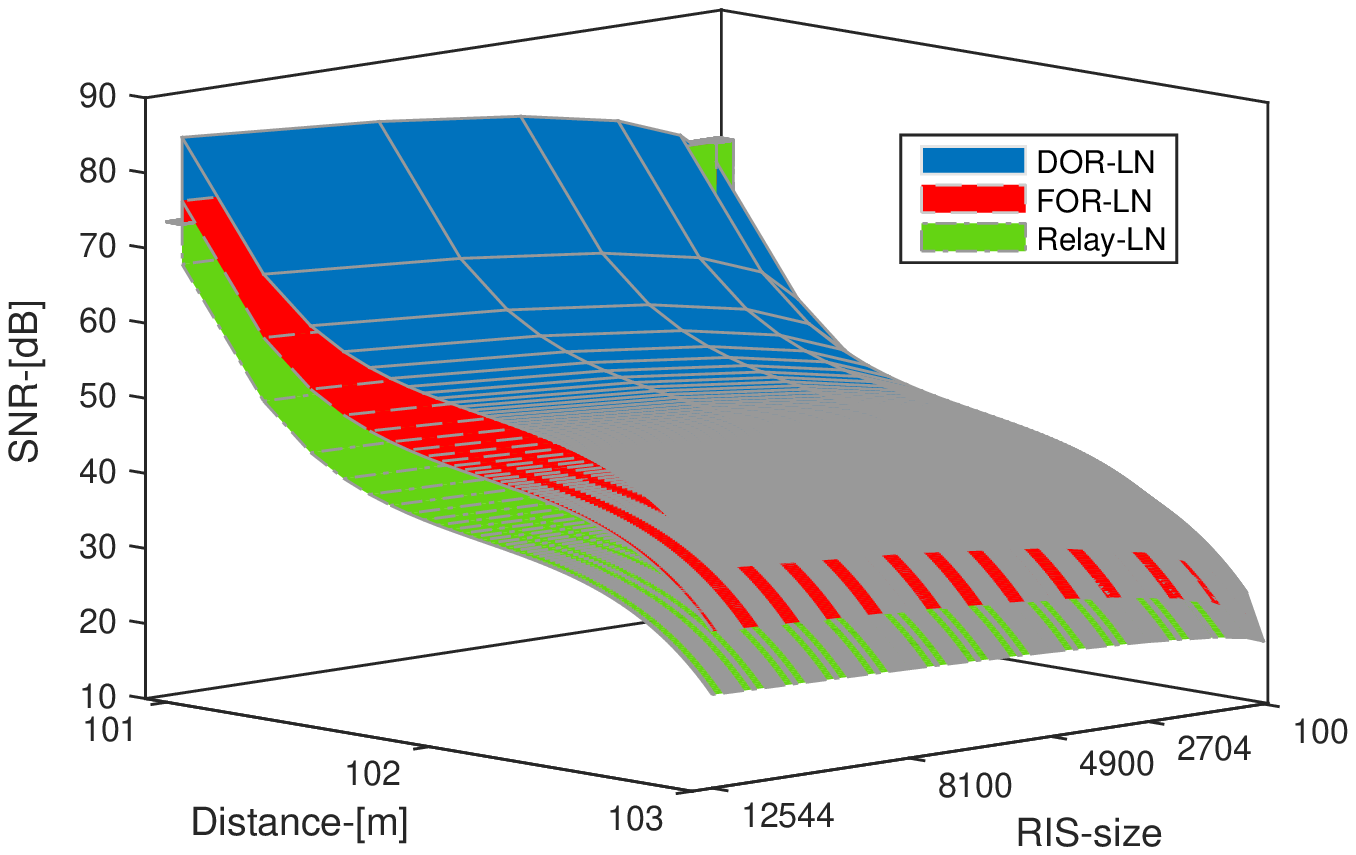}{
\psfrag{SNR-[dB]}{\hspace{-0.72cm}\footnotesize Average SNR [dB]}
\psfrag{Distance-[m]}{\hspace{0.1cm}\footnotesize $L_b\!+\!L_m\,[\text{m}]$}
\psfrag{RIS-size}{\hspace{-0.35cm}\footnotesize RIS size}
\psfrag{FOR-LN}{\hspace{-0.05cm}\scriptsize $\text{FOR}_\text{LN}$}
\psfrag{FOR-GG}{\hspace{-0.05cm}\scriptsize $\text{FOR}_\text{GG}$}
\psfrag{DOR-LN}{\hspace{-0.05cm}\scriptsize $\text{DOR}_\text{LN}$}
\psfrag{DOR-GG}{\hspace{-0.05cm}\scriptsize $\text{DOR}_\text{GG}$}
\psfrag{Relay-LN}{\hspace{-0.05cm}\scriptsize $\text{Relay}_\text{LN}$}
\psfrag{Relay-GG}{\hspace{-0.05cm}\scriptsize $\text{Relay}_\text{GG}$}

\psfrag{10}{\hspace{-0.05cm}\scriptsize $10$}
\psfrag{20}{\hspace{-0.05cm}\scriptsize $20$}
\psfrag{30}{\hspace{-0.05cm}\scriptsize $30$}
\psfrag{40}{\hspace{-0.05cm}\scriptsize $40$}
\psfrag{50}{\hspace{-0.05cm}\scriptsize $50$}
\psfrag{60}{\hspace{-0.05cm}\scriptsize $60$}
\psfrag{70}{\hspace{-0.05cm}\scriptsize $70$}
\psfrag{80}{\hspace{-0.05cm}\scriptsize $80$}
\psfrag{90}{\hspace{-0.05cm}\scriptsize $90$}

\psfrag{101}{\hspace{-0.03cm}\scriptsize $10^1$}
\psfrag{102}{\hspace{-0.08cm}\scriptsize $10^2$}
\psfrag{103}{\hspace{-0.08cm}\scriptsize $10^3$}

\psfrag{100}{\hspace{-0.18cm}\scriptsize $100$}
\psfrag{2704}{\hspace{-0.13cm}\scriptsize $2704$}
\psfrag{4900}{\hspace{-0.16cm}\scriptsize $4900$}
\psfrag{8100}{\hspace{-0.13cm}\scriptsize $8100$}
\psfrag{12544}{\hspace{-0.13cm}\scriptsize $12544$}
}
}
\vspace{-0.3cm}
\hfill
\subfloat[]{
\pstool[scale=0.55]{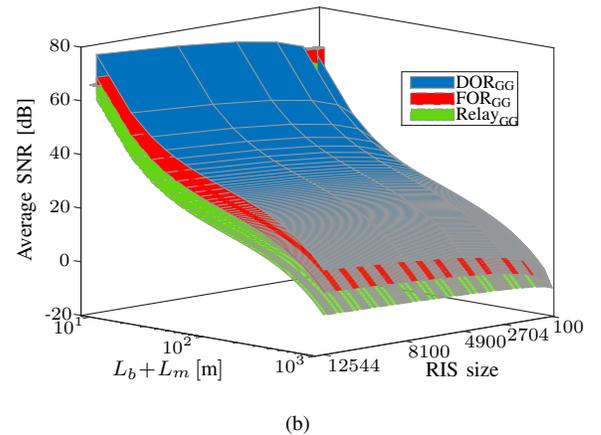}{
\psfrag{SNR-[dB]}{\hspace{-0.72cm}\footnotesize Average SNR [dB]}
\psfrag{Distance-[m]}{\hspace{0.1cm}\footnotesize $L_b\!+\!L_m\,[\text{m}]$}
\psfrag{RIS-size}{\hspace{-0.35cm}\footnotesize RIS size}
\psfrag{FOR-LN}{\hspace{-0.05cm}\scriptsize $\text{FOR}_\text{GG}$}
\psfrag{DOR-LN}{\hspace{-0.05cm}\scriptsize $\text{DOR}_\text{GG}$}
\psfrag{Relay-LN}{\hspace{-0.05cm}\scriptsize $\text{Relay}_\text{GG}$}

\psfrag{-20}{\hspace{-0.05cm}\scriptsize -$20$}
\psfrag{0}{\hspace{-0.025cm}\scriptsize $0$}
\psfrag{20}{\hspace{-0.05cm}\scriptsize $20$}
\psfrag{40}{\hspace{-0.05cm}\scriptsize $40$}
\psfrag{60}{\hspace{-0.05cm}\scriptsize $60$}
\psfrag{80}{\hspace{-0.05cm}\scriptsize $80$}

\psfrag{101}{\hspace{-0.03cm}\scriptsize $10^1$}
\psfrag{102}{\hspace{-0.08cm}\scriptsize $10^2$}
\psfrag{103}{\hspace{-0.08cm}\scriptsize $10^3$}

\psfrag{100}{\hspace{-0.18cm}\scriptsize $100$}
\psfrag{2704}{\hspace{-0.13cm}\scriptsize $2704$}
\psfrag{4900}{\hspace{-0.16cm}\scriptsize $4900$}
\psfrag{8100}{\hspace{-0.13cm}\scriptsize $8100$}
\psfrag{12544}{\hspace{-0.13cm}\scriptsize $12544$}
}}
\caption{The average SNR versus the coverage diameter and RIS size under (a) weak turbulence and (b) moderate-to-strong turbulence FSO channels.}
\label{Fig:Fig.8}
\end{figure}
Fig.~\ref{Fig:Fig.8} depicts the average SNR versus the coverage diameter and a new analysis dimension, the RIS size. It is revealed that increasing the number of RIS elements, i.e., $N_k\!\times\!N_\ell\!\times\!N_m$, would monotonically raise the average SNR for a given coverage diameter. However, the average SNR becomes saturated for gigantic RIS sizes. For plotting this figure, it is assumed that the divergence angle is adjusted based on the RIS size. The reason is that the incident FSO beam from every FSO-BS must cover the whole area of the paired RIS. In other words, the larger RIS sizes, the wider divergence angles, hence more expensive FSO expanders are required.

\subsection{System Design Framework}
A system design framework for practical deployments is suggested here by taking advantage of the provided numerical results and discussions. To this end, we set the outage probability at the HST receiver to be $10^{-3}$. According to Fig.~\ref{Fig:Fig.5}, the required SNR value for each scenario and turbulence condition is deduced. Then, the diameter of the coverage area supported by each FSO-BS and the achievable data rates at the HST are found in Fig.~\ref{Fig:Fig.6} and Fig.~\ref{Fig:Fig.7}, respectively. Table~\ref{Tab:Tab.2} summarizes the proposed design parameters for the assumed railway length of $100\,[\text{km}]$. 
\noindent
\renewcommand{\arraystretch}{1}
\begin{table}[t!]
\begin{center}
\caption{A System Design Sample.}\label{Tab:Tab.2}
\begin{tabular}{| l | c | c | c | c | c | c |}
\cline{2-7}
\multicolumn{1}{c|}{\small}& \multicolumn{2}{c|}{\small FOR} & \multicolumn{2}{c|}{\small DOR} & \multicolumn{2}{c|}{\small Relay}\\
\cline{2-7}
\multicolumn{1}{c|}{\small} & \footnotesize LN & \footnotesize GG & \footnotesize LN & \footnotesize GG & \footnotesize LN & \footnotesize GG\\
\hline
\hline
\multicolumn{1}{|l|}{\footnotesize \hspace{-0.1cm} $P_{out}$} & \multicolumn{6}{c|}{\footnotesize $10^{-3}$}\\
\hline
\footnotesize \hspace{-0.1cm} $\text{Required SNR}\,[\text{dB}]$ & \footnotesize $25$ & \footnotesize $35$ & \footnotesize $22$ & \footnotesize $33$ & \footnotesize $27$ & \footnotesize $37$\\
\hline
\footnotesize \hspace{-0.1cm} $L_b\!+\!L_m\,[\text{m}]$ & \footnotesize $\!\!1000\!\!$ & \footnotesize $28$ & \footnotesize $\!\!1130\!\!$ & \footnotesize $88$ & \footnotesize $325$ & \footnotesize $20$\\
\hline
\footnotesize \hspace{-0.1cm} Number of FSO-BSs ($B$)\scriptsize*& \footnotesize $100$ & \footnotesize $\!\!3572\!\!$ & \footnotesize $89$ & \footnotesize $114$ & \footnotesize $308$ & \footnotesize $\!\!5000$\!\!\\
\hline
\multicolumn{7}{l}{\scriptsize *It is assumed that the length of the railway line is $100\,[\text{km}]$.}\\
\end{tabular}
\medskip
\end{center}
\end{table}

\begin{figure}[t!]
\centering
\pstool[scale=0.55]{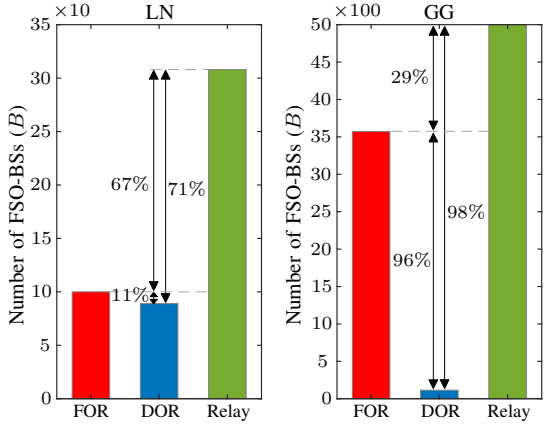}{
\psfrag{Number-of-FSO-BSs(B)}{\hspace{-0.51cm}\footnotesize Number of FSO-BSs ($B$)}
\psfrag{Distance-[m]}{\hspace{-0.1cm}\footnotesize $L_b\!+\!L_m \,[\text{m}]$}
\psfrag{DIrect-FSO-connections}{\hspace{-0.05cm}\scriptsize Direct FSO connections}
\psfrag{FOR}{\hspace{-0.06cm}\scriptsize FOR}
\psfrag{DOR}{\hspace{-0.06cm}\scriptsize DOR}
\psfrag{Relay}{\hspace{-0.06cm}\scriptsize Relay}
\psfrag{LN}{\hspace{-0.06cm}\footnotesize LN}
\psfrag{GG}{\hspace{-0.06cm}\footnotesize GG}
\psfrag{0}{\hspace{-0.06cm}\scriptsize $0$}
\psfrag{50}{\hspace{0.025cm}\scriptsize $5$}
\psfrag{100}{\hspace{-0.0cm}\scriptsize $10$}
\psfrag{150}{\hspace{-0.0cm}\scriptsize $15$}
\psfrag{200}{\hspace{-0.0cm}\scriptsize $20$}
\psfrag{250}{\hspace{-0.0cm}\scriptsize $25$}
\psfrag{300}{\hspace{-0.0cm}\scriptsize $30$}
\psfrag{350}{\hspace{-0.0cm}\scriptsize $35$}

\psfrag{500}{\hspace{0.0925cm}\scriptsize $5$}
\psfrag{1000}{\hspace{0.09cm}\scriptsize $10$}
\psfrag{1500}{\hspace{0.09cm}\scriptsize $15$}
\psfrag{2000}{\hspace{0.09cm}\scriptsize $20$}
\psfrag{2500}{\hspace{0.09cm}\scriptsize $25$}
\psfrag{3000}{\hspace{0.09cm}\scriptsize $30$}
\psfrag{3500}{\hspace{0.09cm}\scriptsize $35$}
\psfrag{4000}{\hspace{0.09cm}\scriptsize $40$}
\psfrag{4500}{\hspace{0.09cm}\scriptsize $45$}
\psfrag{5000}{\hspace{0.09cm}\scriptsize $50$}

\psfrag{*10}{\hspace{0.09cm}\scriptsize $\times10$}
\psfrag{*100}{\hspace{0.09cm}\scriptsize $\times100$}

\psfrag{11d}{\hspace{-0.11cm}\scriptsize $11\%$}
\psfrag{67.5d}{\hspace{-0.0cm}\scriptsize $67\%$}
\psfrag{71d}{\hspace{-0.05cm}\scriptsize $71\%$}
\psfrag{96d}{\hspace{-0.135cm}\scriptsize $96\%$}
\psfrag{98d}{\hspace{-0.042cm}\scriptsize $98\%$}
\psfrag{28.5d}{\hspace{0.012cm}\scriptsize $29\%$}
}
\vspace{-0.2cm}
\caption{The number of the required FSO-BSs for the assumed railway length of $100\,[\text{km}]$.}
\label{Fig:Fig.9}
\end{figure}
Fig.~\ref{Fig:Fig.9} shows the number of needed FSO-BSs for the assumed railway setup under different scenarios and turbulence conditions, as given in Table~\ref{Tab:Tab.2}. It is verified that the DOR strategy requires much simpler infrastructures with less handover frequency compared to the other systems, at the cost of demanding a train tracking procedure. In addition, the FOR strategy requires fewer FSO-BSs, hence fewer handover processes, with respect to the traditional relay system.

\section{Conclusion}\label{Sec:Sec5}
In this paper, we proposed a novel FSO access setup based on the RIS technology to enhance the performance of the existing access networks which provide broadband Internet access for HSTs.
We firstly derived the statistical expressions of direct and RIS-assisted FSO channels under weak and moderate-to-strong fading conditions. Then, the average SNR and outage probability formulas at the HST were computed for the FOR and DOR coverage scenarios.
Through the numerical results, it was shown that the proposed access network offers up to around $44\%$ higher data rates and $240\%$ wider coverage area served by each FSO-BS compared to the relay-assisted alternative. 
The increase of the FSO-BS coverage on average reduces $67\%$ the number of FSO-BSs required to support a given railway distance. Decreasing the number of FSO-BSs in an access network would decrease the handover frequency and CAPEX, as well.
The main results were verified by Monte-Carlo simulations.
Finally, a system design framework for practical deployments was suggested based on the achieved results and assumed design sample. The calculated values on the design table verified our results that the proposed setup boosts the performances of the existing alternatives.

\appendices
\section{Proof of Theorem \ref{The:The1}} \label{App:App1}
Using (\ref{Eq:Eq4}) and \cite[Eq.~(12)]{farid2007outage}, the p.d.f of $h_{k\ell}$ is obtained as
\begin{align}\label{Eq:EqApp1} \nonumber
    &f_{h_{k\ell}}(h_{k\ell}) \!=\! \int\! f_{h_{k\ell}|h_{t,k\ell}}\!\left(h_{k\ell}|h_{t,k\ell}\right)f_{h_{t,k\ell}}(h_{t,k\ell})\, dh_{t,k\ell}\\ \nonumber
    &=\! \int_{\frac{h_{k\ell}}{\mathcal{C}_a}}^{\infty} \frac{1}{h_{p,k\ell}h_{t,k\ell}} f_{h_{g,k\ell}}\!\!\left(\!\frac{h_{k\ell}}{h_{p,k\ell}h_{t,k\ell}}\!\right)f_{h_{t,k\ell}}(h_{t,k\ell})\, dh_{t,k\ell}\\ \nonumber
    &=\! \dfrac{\sqrt{\varpi_{k\ell}}}{2\sqrt{\pi}h_{k\ell}}\left(\! \dfrac{h_{k\ell}}{\mathcal{C}_a} \!\right)^{\!\!\varpi_{k\ell}} \\
    &~~\times\! \int_{\frac{h_{k\ell}}{\mathcal{C}_a}}^{\infty} \left[ \ln\!\left(\! \dfrac{\mathcal{C}_ah_{t,k\ell}}{h_{k\ell}} \!\right) \!\right]^{\!-1/2} \!\!\dfrac{1}{h_{t,k\ell}^{\varpi_{k\ell}}} f_{h_{t,k\ell}}(h_{t,k\ell})\, dh_{t,k\ell},
\end{align}
where $\mathcal{C}_a\!=\!h_{0,k\ell}h_{p,k\ell}$. Now, by inserting (\ref{Eq:Eq2}) for the weak turbulence into (\ref{Eq:EqApp1}) and some mathematical manipulations, (\ref{Eq:Eq6}) is achieved.

\section{Proof of Theorem \ref{The:The2}} \label{App:App2}
By applying the same technique as in (\ref{Eq:EqApp1}) and using (\ref{Eq:Eq2}) for the moderate-to-strong turbulence, we have
\begin{subequations}
\begin{align} \nonumber
    &f_{h_{k\ell}}(h_{k\ell})
    \!=\! \dfrac{\sqrt{\varpi_{k\ell}}}{\sqrt{\pi}h_{k\ell}} \dfrac{(\alpha_{k\ell}\beta_{k\ell})^\frac{\alpha_{k\ell}\!+\!\beta_{k\ell}}{2}}{\Gamma(\alpha_{k\ell})\Gamma(\beta_{k\ell})} \left(\! \dfrac{h_{k\ell}}{\mathcal{C}_a} \!\right)^{\!\!\varpi_{k\ell}} \\ \nonumber
    &~~\times\! \int_{\frac{h_{k\ell}}{\mathcal{C}_a}}^{\infty} \left[ \ln\!\left(\! \dfrac{\mathcal{C}_ah_{t,k\ell}}{h_{k\ell}} \!\right) \!\right]^{\!-1/2} h_{t,k\ell}^{-\!\varpi_{k\ell}\!-\!1} \\ \label{Eq:EqApp2a}
    &~~\times{h_{t,k\ell}^{\frac{1}{2}({\alpha_{k\ell}\!+\!\beta_{k\ell}})}}  K_{\alpha_{k\ell}\!-\!\beta_{k\ell}}\!\Big(\!2 \sqrt{\alpha_{k\ell}\beta_{k\ell}h_{t,k\ell}}\Big) dh_{t,k\ell}\\ \nonumber
    &=\! \dfrac{\sqrt{\varpi_{k\ell}}}{2\sqrt{\pi}h_{k\ell}} \dfrac{(\alpha_{k\ell}\beta_{k\ell})^{\varpi_{k\ell}}}{\Gamma(\alpha_{k\ell})\Gamma(\beta_{k\ell})} \left(\! \dfrac{h_{k\ell}}{\mathcal{C}_a} \!\right)^{\!\!\varpi_{k\ell}} \\ \nonumber
    &~~\times\! \int_{\frac{\alpha_{k\ell}\beta_{k\ell}h_{k\ell}}{\mathcal{C}_a}}^{\infty} \left[ \ln\!\left(\! \dfrac{\mathcal{C}_a h_{t,k\ell}}{\alpha_{k\ell}\beta_{k\ell}h_{k\ell}} \!\right) \!\right]^{\!-1/2} h_{t,k\ell}^{-\!\varpi_{k\ell}\!-\!1} \\ \label{Eq:EqApp2c}
    &~~\times \MeijerG*{2}{0}{0}{2}{0}{\alpha_{k\ell}, \beta_{k\ell}\!}{\!h_{t,k\ell}} dh_{t,k\ell}\\ \nonumber
    &\cong\! \dfrac{\sqrt{\varpi_{k\ell} {\zeta}}}{2\sqrt{\pi}h_{k\ell}} \dfrac{(\alpha_{k\ell}\beta_{k\ell})^{\varpi_{k\ell}}}{\Gamma(\alpha_{k\ell})\Gamma(\beta_{k\ell})} \left(\! \dfrac{h_{k\ell}}{\mathcal{C}_a} \!\right)^{\!\!\varpi_{k\ell}} \left(\!\dfrac{\alpha_{k\ell}\beta_{k\ell}h_{k\ell}}{\mathcal{C}_a}\!\right)^{\!\!1/{2\zeta}\!-\!1} \\ \nonumber
    &~~\times\! \int_{\frac{\alpha_{k\ell}\beta_{k\ell}h_{k\ell}}{\mathcal{C}_a}}^{\infty} \left[ h_{t,k\ell} \!-\!  \dfrac{\alpha_{k\ell}\beta_{k\ell}h_{k\ell}}{\mathcal{C}_a} \!\right]^{\!-1/{2\zeta}} h_{t,k\ell}^{-\!\varpi_{k\ell}\!-\!1} \\ \label{Eq:EqApp2d}
    &~~\times \MeijerG*{2}{0}{0}{2}{0}{\alpha_{k\ell}, \beta_{k\ell}\!}{\!h_{t,k\ell}} dh_{t,k\ell},
    \end{align}
\end{subequations}
where $\mathcal{C}_c\!=\!1/{2\zeta}\!+\!\varpi_{k\ell}\!-\!1$. Next, (\ref{Eq:Eq7}) is obtained after some mathematical computations. Note that, to derive (\ref{Eq:EqApp2c}) and solve (\ref{Eq:EqApp2d}), the similar steps as in \cite[Eq.~(5)]{alheadary2017ber} and \cite[Eq.~(6)]{alheadary2017ber} are passed, respectively. In addition, to derive (\ref{Eq:EqApp2d}), we utilize the approximation $\ln(x)\!\cong\!\zeta x^{1/\zeta}\!-\!\zeta$, for the large enough $\zeta$.

\section{Proof of Theorem \ref{The:The3}} \label{App:App3}
Similar to (\ref{Eq:EqApp1}), the p.d.f of $h_{b}(t)$ is obtained as
\begin{subequations}
\begin{align}\label{Eq:EqApp3} \nonumber
    &f_{h_{b}(t)}(h_{b}(t)) \!=\! \frac{\xi^2(h_b(t))^{\xi^2-1}}{\mathcal{C}_d^{\xi^2}}\\ 
    &~~\times\!  \int_{\frac{h_b(t)}{\mathcal{C}_d}}^\infty\!   ({h_{t,b}(t)})^{-\xi^2} f_{h_{t,b}(t)}(h_{t,b}(t))\, dh_{t,b}(t)\\ \nonumber
    &=\!\frac{\xi^2(h_b(t))^{\xi^2-1}}{\mathcal{C}_d^{\xi^2}} \dfrac{2(\alpha_{b}\beta_{b})^{\frac{1}{2}({\alpha_{b}\!+\!\beta_{b}})}}{\Gamma(\alpha_{b})\Gamma(\beta_{b})} \int_{\frac{h_b(t)}{\mathcal{C}_d}}^\infty\! ({h_{t,b}(t)})^{-\xi^2\!-\!1} \\ 
    &~~\times\! ({h_{t,b}(t)})^{\frac{1}{2}({\alpha_{b}\!+\!\beta_{b}})} K_{\alpha_{b}\!-\!\beta_{b}}\!\Big(\!2 \sqrt{\!\alpha_{b}\beta_{b}h_{t,b}(t)}\Big)\, dh_{t,b}(t)\\ \nonumber
    &=\!\frac{\xi^2(h_b(t))^{\xi^2-1}}{\mathcal{C}_d^{\xi^2}} \dfrac{(\alpha_{b}\beta_{b})^{\xi^2}}{\Gamma(\alpha_{b})\Gamma(\beta_{b})} \\ 
    &~~\times\! \int_{\frac{\alpha_{b}\beta_{b}h_b(t)}{\mathcal{C}_d}}^\infty\! ({h_{t,b}(t)})^{-\xi^2\!-\!1} \MeijerG*{2}{0}{0}{2}{0}{\alpha_{b}, \beta_{b}\!}{\!h_{t,b}(t)}\, dh_{t,b}(t).
\end{align}
\end{subequations}
Then, (\ref{Eq:Eq11}) is achieved after some mathematical manipulations.

\section{Proof of Theorem \ref{The:The4}} \label{App:App4}
We firstly derive the p.d.f of $\gamma_b(t)$ then compute its corresponding CDF. In this regard, we define $\gamma_{m,k\ell}\!\triangleq\!\bar{\gamma}_{k\ell} h_{m,k\ell}^2$ and $\gamma_{m,k\ell}^\prime\!\triangleq\!\bar{\gamma}^\prime_{k\ell} h_{m,k\ell}^2$. Besides, we know that the RIS-assisted channels are uncorrelated, hence independent. Thus, we have 
\begin{align}\label{Eq:EqApp5} \nonumber
    &f_{\gamma_b(t)}(\gamma_b(t)) \!=\! f_{\bar{\gamma}_b (h_b(t))^2}(\bar{\gamma}_b (h_b(t))^2) \\
    &~~~=\! \frac{\xi^2 (\gamma_b(t))^{\frac{\xi^2}{2}-1}\,\mathcal{C}_f}{4 \bar{\gamma}_b^{\frac{\xi^2}{2}} \mathcal{C}_d^{\xi^2}} \operatorname{erfc}\! \left(\! \frac{0.5\ln\!\left(\!{\frac{\gamma_b(t)}{\bar{\gamma}_b \mathcal{C}_d^2}}\!\right) \!+\! \mathcal{C}_e}{\sqrt{8 \sigma_{b}^2}}\!\right) \!\mathbb{I}(\mathcal{C}_1),
\end{align}
for the direct FSO channels. Likewise, we have
\begin{align}\label{Eq:EqApp6} \nonumber
    &f_{\gamma_b(t)}(\gamma_b(t)) \!=\! \prod\limits_{\ell=1}^{N_\ell}\prod\limits_{k=1}^{N_k} f_{{\gamma}_{m,k\ell}}({\gamma}_{m,k\ell}) \\ \nonumber
    &=\! \prod\limits_{\ell=1}^{N_\ell}\prod\limits_{k=1}^{N_k} \dfrac{\sqrt{\varpi_{k\ell}}\, \gamma_{m,k\ell}^{\frac{\varpi_{k\ell}}{2}-1}}{16\sqrt{\pi} \bar{\gamma}_{k\ell}^\frac{\varpi_{k\ell}}{2}\mathcal{C}_a^{\varpi_{k\ell}}} \operatorname{erfc} \!\left(\! \frac{0.5\ln\!\left(\!\frac{\gamma_{m,k\ell}}{\bar{\gamma}_{k\ell} \mathcal{C}_a^2} \!\right)\!+\!\mathcal{C}_b}{\sqrt{8 \sigma_{k\ell}^2}} \!\right)\\ \nonumber
    &~~\times\! \Bigg\{\!\! \left[\dfrac{\big(5\!-\!\ln(\mathcal{C}_a)\big) \bar{\gamma}_{k\ell}^\frac{\varpi_{k\ell}}{2}}{\gamma_{m,k\ell}^{ \frac{\varpi_{k\ell}}{2}}}\!+\!\dfrac{\sqrt{8\sigma_{k\ell}^2}}{\sqrt{\pi}}\right] \\ \nonumber
    &~~\times\! \operatorname{exp}\!\left(\!\dfrac{\big(0.5\ln\! \left(\!\frac{\gamma_{m,k\ell}}{\bar{\gamma}_{k\ell}}\!\right)\!+\!\mathcal{C}_b\big)^2 \!-\! \big(0.5\ln\!\left(\!\frac{\gamma_{m,k\ell}}{\bar{\gamma}_{k\ell}}\!\right)\!+\!2\sigma_{k\ell}^2\big)^2   }{8\sigma_{k\ell}^2} \!\right)\\
    &~~+\! \left(\!\ln\!\left(\!\frac{\gamma_{m,k\ell}}{\bar{\gamma}_{k\ell}}\!\right)\!+\!\mathcal{C}_b\!\right) \operatorname{exp}\!\Big(\!2\sigma_{k\ell}^2\varpi_{k\ell}(1\!+\!\varpi_{k\ell})\!\Big) \!\!\Bigg\} \mathbb{I}(\mathcal{C}_2),
\end{align}
for the FOR RIS-assisted FSO channels, while
\begin{align}\label{Eq:EqApp7} \nonumber
    &f_{\gamma_b(t)}(\gamma_b(t)) \!=\!\!\! \prod\limits_{n_m=1}^{N_m}\prod\limits_{\ell=1}^{N_\ell}\prod\limits_{k=1}^{N_k} f_{{\gamma}^\prime_{m,k\ell}}({\gamma}^\prime_{m,k\ell}) \\ \nonumber
    &=\!\!\! \prod\limits_{n_m=1}^{N_m} \prod\limits_{\ell=1}^{N_\ell}\prod\limits_{k=1}^{N_k} \dfrac{\sqrt{\varpi_{k\ell}}\, \gamma_{m,k\ell}^{\prime \frac{\varpi_{k\ell}}{2}-1}}{16\sqrt{\pi} \bar{\gamma}_{k\ell}^{\prime\frac{\varpi_{k\ell}}{2}}\mathcal{C}_a^{\varpi_{k\ell}}} \operatorname{erfc} \!\left(\! \frac{0.5\ln\!\left(\!\frac{\gamma_{m,k\ell}^\prime}{\bar{\gamma}^\prime_{k\ell} \mathcal{C}_a^2} \!\right)\!+\!\mathcal{C}_b}{\sqrt{8 \sigma_{k\ell}^2}} \!\right)\\ \nonumber
    &~~\times\! \Bigg\{\!\! \left[\dfrac{\big(5\!-\!\ln(\mathcal{C}_a)\big) \bar{\gamma}_{k\ell}^{\prime\frac{\varpi_{k\ell}}{2}}}{\gamma_{m,k\ell}^{\prime \frac{\varpi_{k\ell}}{2}}}\!+\!\dfrac{\sqrt{8\sigma_{k\ell}^2}}{\sqrt{\pi}}\right] \\ \nonumber
    &~~\times\! \operatorname{exp}\!\left(\!\dfrac{\big(0.5\ln\! \left(\!\frac{\gamma_{m,k\ell}^\prime}{\bar{\gamma}^\prime_{k\ell}}\!\right)\!+\!\mathcal{C}_b\big)^2 \!-\! \big(0.5\ln\!\left(\!\frac{\gamma_{m,k\ell}^\prime}{\bar{\gamma}^\prime_{k\ell}}\!\right)\!+\!2\sigma_{k\ell}^2\big)^2   }{8\sigma_{k\ell}^2} \!\right)\\
    &~~+\! \left(\!\ln\!\left(\!\frac{\gamma_{m,k\ell}^\prime}{\bar{\gamma}^\prime_{k\ell}}\!\right)\!+\!\mathcal{C}_b\!\right) \operatorname{exp}\!\Big(\!2\sigma_{k\ell}^2\varpi_{k\ell}(1\!+\!\varpi_{k\ell})\!\Big) \!\!\Bigg\} \mathbb{I}(\mathcal{C}_3), 
\end{align}
for the DOR RIS-assisted FSO channels. In (\ref{Eq:EqApp5})--(\ref{Eq:EqApp7}), we utilize (\ref{Eq:Eq6}), (\ref{Eq:Eq10}), and (\ref{Eq:Eq13}).
Now, (\ref{Eq:Eq17})--(\ref{Eq:Eq19}) are obtained from (\ref{Eq:EqApp5})--(\ref{Eq:EqApp7}) pairwise by the use of the conventional CDF formula and the linearity property of integrals.

\section{Proof of Theorem \ref{The:The5}} \label{App:App5}
Similar to Appendix~\ref{App:App4}, the p.d.f expressions are as follows
\begin{align} \label{Eq:EqApp8} \nonumber
    f_{\gamma_b(t)}(\gamma_b(t)) &\!=\! \dfrac{\xi^2 (\gamma_b(t))^{\frac{\xi^2}{2}-1}}{2\bar{\gamma}_b^\frac{\xi^2}{2} \Gamma(\alpha_{b})\Gamma(\beta_{b})} \!\left(\!\dfrac{\alpha_{b}\beta_{b}}{\mathcal{C}_d}\!\right)^{\!\!\xi^2}\\ 
    &~~\times\! \MeijerG*{3}{0}{1}{3}{1+\xi^2}{\xi^2, \alpha_{b}, \beta_{b}\!}{\!\frac{\alpha_{b}\beta_{b}\sqrt{\gamma_{b}(t)}}{\sqrt{\bar{\gamma}_{b}}\,\mathcal{C}_d}} \mathbb{I}(\mathcal{C}_1),
\end{align}
for the direct FSO channels. Besides, we have
\begin{align} \label{Eq:EqApp9} \nonumber
    &f_{\gamma_b(t)}(\gamma_b(t)) \!=\! \prod\limits_{\ell=1}^{N_\ell}\prod\limits_{k=1}^{N_k} \dfrac{\sqrt{\varpi_{k\ell} {\zeta}} \gamma_{m,k\ell}^{\frac{\mathcal{C}_c}{2}-1}}{2\sqrt{\pi} \bar{\gamma}_{k\ell}^{\frac{\mathcal{C}_c}{2}} \Gamma(\alpha_{k\ell})\Gamma(\beta_{k\ell})}  \!\left(\!\dfrac{\alpha_{k\ell}\beta_{k\ell}}{\mathcal{C}_a}\!\right)^{\!\!\mathcal{C}_c} \\
    &~~~~~~~~~\times\!
    \MeijerG*{3}{0}{1}{3}{1+\varpi_{k\ell}}{1+\mathcal{C}_c, \alpha_{k\ell}, \beta_{k\ell}\!}{\!\frac{\alpha_{k\ell}\beta_{k\ell}\sqrt{\gamma_{m,k\ell}}}{\sqrt{\bar{\gamma}_{k\ell}}\,\mathcal{C}_a}} \mathbb{I}(\mathcal{C}_2),
\end{align}
for the FOR RIS-assisted FSO channels, and
\begin{align} \label{Eq:EqApp10} \nonumber
    &\!\!f_{\gamma_b(t)}(\gamma_b(t)) \!=\!\!\! \prod\limits_{n_m=1}^{N_m}\prod\limits_{\ell=1}^{N_\ell}\prod\limits_{k=1}^{N_k} \dfrac{\sqrt{\varpi_{k\ell} {\zeta}} \gamma_{m,k\ell}^{\prime \frac{\mathcal{C}_c}{2}-1}}{2\sqrt{\pi} \bar{\gamma}_{k\ell}^{\prime \frac{\mathcal{C}_c}{2}} \Gamma(\alpha_{k\ell})\Gamma(\beta_{k\ell})} \! \left(\!\dfrac{\alpha_{k\ell}\beta_{k\ell}}{\mathcal{C}_a}\!\right)^{\!\!\mathcal{C}_c}\\
    &\!\!~~~~~~~~~\times
    \!\MeijerG*{3}{0}{1}{3}{1+\varpi_{k\ell}}{1+\mathcal{C}_c, \alpha_{k\ell}, \beta_{k\ell}\!}{\!\frac{\alpha_{k\ell}\beta_{k\ell}\sqrt{\gamma_{m,k\ell}^\prime}}{\sqrt{\bar{\gamma}_{k\ell}^\prime}\,\mathcal{C}_a}} \mathbb{I}(\mathcal{C}_3),
\end{align}
for the DOR RIS-assisted FSO channels. Through (\ref{Eq:EqApp8})--(\ref{Eq:EqApp10}), we apply (\ref{Eq:Eq7}), (\ref{Eq:Eq11}), and (\ref{Eq:Eq13}).
Next, (\ref{Eq:Eq20})--(\ref{Eq:Eq22}) are obtained from (\ref{Eq:EqApp8})--(\ref{Eq:EqApp10}), respectively.

\balance
\bibliographystyle{IEEEtran}
\bibliography{References.bib}

\end{document}